\documentclass[preprint,authoryear]{elsarticle}

\usepackage{graphicx}
\usepackage{amssymb,amsthm,amsmath}
\usepackage{mathtools}
\usepackage{booktabs}
\usepackage{paralist}
\usepackage{tikz-cd}
\usepackage{tikz-network}
\usepackage{caption}
\usepackage{subcaption}
\usepackage{mathdots}
\usepackage{csquotes}
\usepackage[linesnumbered,ruled,vlined,noend]{algorithm2e}
\usepackage{todonotes}
\usepackage{xargs}
\usepackage{comment}
\usepackage{standalone}
\usepackage{sgame}
\usepackage{bm}

\usepackage{subcaption}

\tikzset{ball/.style={circle, draw, fill=black,inner sep=0pt, minimum width=4pt}}

\usepackage{pgfplots}
\pgfplotsset{compat = newest}

\usepackage{xcolor}
\usetikzlibrary{decorations.markings}

\usepackage{tikz}
\usetikzlibrary{arrows.meta}
\tikzset{label/.style = {inner sep=1pt, fill=white}}
\tikzset{nd/.style={inner sep=1pt}}
\tikzset{>=Latex}
\tikzset{arc/.style = {->, semithick, >=Latex}}

\theoremstyle{plain}
\newtheorem{thm}{Theorem}[section]

\theoremstyle{definition}
\newtheorem{defn}[thm]{Definition}

\newtheorem{lem}[thm]{Lemma}

\newtheorem{rmk}{Remark}
\newtheorem{conj}[thm]{Conjecture}
\newtheorem{prop}[thm]{Proposition}

\newcommand{\real}{\mathbb{R}}

\newcommand{\exptu}{\mathbb{U}}

\newcommand{\arc}[3][]{ #2 \xlongrightarrow{#1} #3}
\newcommand{\edge}[3][]{#2 \xleftrightarrow[]{#1} #3}

\newcommand{\toprove}[1]{\todo[color=green]{#1}}

\title{Preference graphs: a combinatorial tool for game theory}
\author{Oliver Biggar and Iman Shames}
\date{}

\begin{document}

\begin{abstract}
    The \emph{preference graph} is a combinatorial representation of the structure of a normal-form game. Its nodes are the strategy profiles, with an arc between profiles if they differ in the strategy of a single player, where the orientation indicates the preferred choice for that player. We show that the preference graph is a surprisingly fundamental tool for studying normal-form games, which arises from natural axioms and which underlies many key game-theoretic concepts, including \emph{dominated strategies} and \emph{strict Nash equilibria}, as well as classes of games like \emph{potential games}, \emph{supermodular games} and \emph{weakly acyclic games}. The preference graph is especially related to \emph{game dynamics}, playing a significant role in the behaviour of \emph{fictitious play} and the \emph{replicator dynamic}. Overall, we aim to equip game theorists with the tools and understanding to apply the preference graph to new problems in game theory.
\end{abstract}

\maketitle

\section{Introduction}
In recent game theory research, one combinatorial object has been an increasing focus of attention. This object---the \emph{preference graph}\footnote{Also called the \emph{response graph} or \emph{better-response graph}. See Section~\ref{sec:nomenclature} for a discussion of naming.} of a normal-form game \citep{biggar_graph_2023,biggar_attractor_2024}---is a simple concept with highly non-trivial connections to the behaviour of dynamics in games. Conceptually, the preference graph is a \emph{combinatorial skeleton} of the game, which describes the underlying structure of the player's preferences~\citep{pangallo_best_2019}.

The preference graph, and its cousin the \emph{best-response graph}, have appeared increasingly often in the literature, playing key roles in a diverse collection of results in algorithmic game theory, machine learning and economics. In algorithmic game theory, the best-response graph was first used to define the Price of Sinking \citep{goemans_sink_2005,mirrokni_convergence_2004}, a non-Nash equilibrium-based alternative to the Price of Anarchy \citep{roughgarden2010algorithmic} which gives very different predictions in the context of game dynamics \citep{kleinberg_beyond_2011}. The preference and best-response graphs were later used in the complexity theory of succinct games, modelling the behaviour of internet routing protocols like BGP \citep{fabrikant2008complexity,mirrokni_complexity_2009}. The preference graph is also the basis of the \emph{harmonic decomposition of games} \citep{candogan_flows_2011}. More recently, the preference graph has been applied to machine learning, with \cite{omidshafiei_-rank_2019,omidshafiei_navigating_2020} using it as a foundation for exploring the space of games and computing the `evolutionary strength' of strategies. \cite{biggar_replicator_2023,biggar_attractor_2024} used the preference graph to analyse the attractors of the \emph{replicator dynamic} \citep{taylor_evolutionary_1978}, providing the first characterisation of the unique attractor of zero-sum games. 
Finally, the preference and best-response graphs have been the subject of three different recent proposals for reevaluating the foundations of game theory and game dynamics \citep{papadimitriou_game_2019,pangallo_best_2019,biggar_graph_2023}.



The definition of the preference graph is simple. The nodes of the graph are the strategy profiles of the game, with an arc between two profiles if they differ in the strategy of a single player. The direction of this arc indicates the preferred option for that player. Sinks of the preference graph are exactly the pure Nash equilibria (PNEs). Consequently, preference graphs have a natural associated solution concept: the \emph{sink strongly connected components}, called the \emph{sink equilibria} \citep{goemans_sink_2005,papadimitriou_game_2019,biggar_graph_2023}.

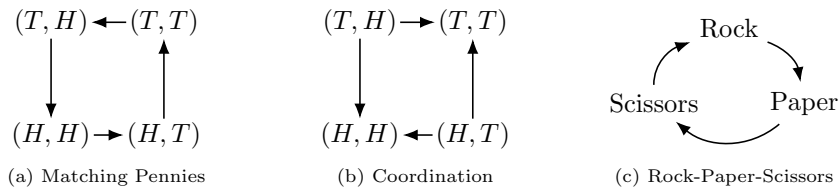
\begin{figure}[h]
    \centering
    \begin{subfigure}{.3\textwidth}
        \centering
        \begin{tikzpicture}
    \node[nd] (A) at (0,0) {$(H,H)$};
    \node[nd] (B) at (0,1.5) {$(T,H)$};
    \node[nd] (C) at (1.5,0) {$(H,T)$};
    \node[nd] (D) at (1.5,1.5) {$(T,T)$};
    \draw[arc] (B) to (A);
    \draw[arc] (A) to (C);
    \draw[arc] (C) to (D);
    \draw[arc] (D) to (B);
    \end{tikzpicture}
        \caption{Matching Pennies}
        \label{fig:MP 1}
    \end{subfigure}
    \quad
    \begin{subfigure}{.3\textwidth}
        \centering
        \begin{tikzpicture}
    \node[nd] (A) at (0,0) {$(H,H)$};
    \node[nd] (B) at (0,1.5) {$(T,H)$};
    \node[nd] (C) at (1.5,0) {$(H,T)$};
    \node[nd] (D) at (1.5,1.5) {$(T,T)$};
    \draw[arc] (B) to (A);
    \draw[arc] (C) to (A);
    \draw[arc] (B) to (D);
    \draw[arc] (C) to (D);
\end{tikzpicture}
        \caption{Coordination}
        \label{fig:CO 1}
    \end{subfigure}
    \quad
    \begin{subfigure}{.3\textwidth}
        \centering
        \begin{tikzpicture}

    \node (rock) at (0,1) {Rock};
    \node (paper) at (1,0) {Paper};
    \node (scissors) at (-1,0) {Scissors};
    
    \draw[arc] (scissors) to[out=90,in=-160] (rock);
    \draw[arc] (paper) to[out=-140,in=-40] (scissors);
    \draw[arc] (rock) to[out=-20,in=100] (paper);
 \end{tikzpicture}
        \caption{Rock-Paper-Scissors}
        \label{fig: RPS 1}
    \end{subfigure}
    \caption{Examples of preference graphs: the $2\times 2$ games Matching Pennies (\ref{fig:MP 1}) and Coordination (\ref{fig:CO 1}), and the symmetric zero-sum game Rock-Paper-Scissors (\ref{fig: RPS 1}). In symmetric zero-sum games the preference graph has a special form, where nodes are strategies and arcs represent the winning strategy in any pair. See Section~\ref{sec: symmetric zero-sum}.}
    \label{fig:examples}
\end{figure}

In Figure~\ref{fig:examples}, we show the preference graphs of some games: Coordination, a $2\times 2$ game where both players aim to match the choice of the other; Matching Pennies, a $2\times 2$ game where one player aims to match, the other to mismatch; and Rock-Paper-Scissors. What is striking about Figure~\ref{fig:examples} is how clearly each preference graph reflects the rules of each game. Preference graphs sidestep a key challenge of payoff matrices \citep{von2007theory}---instead of needing to `invent' numerical payoffs to model these games mathematically, the preference graph emerges directly from their descriptions. 

The preference graph is \emph{so} simple---``almost familiar," in the words of \citeauthor{papadimitriou_game_2019}---that it suggests a fundamental role in game theory, reflected by these varied modern applications. This is the topic of this paper: \textbf{what does the preference graph have to say about game theory?} We explore this idea in two parts: first, we lay out the philosophical motivations and basic properties of the preference graph (Section~\ref{sec: graphs}), and second, we investigate how the preference graph has been used in the past, by examining the literature and rephrasing classical results using the preference graph. We find that the preference graph plays a role---often only implicitly---in many important structural and dynamic results across the history of game theory (Sections~\ref{sec: structure} and \ref{sec: dynamics}). 
These goals---especially the second---are similar to those of a survey. Broadly speaking, our paper \emph{is} a survey, though it is unusual in two ways. First, we do present some new results (such as Lemmas~\ref{shapley uniqueness} and \ref{jordan uniquenesss}) or new proofs of known results, where these provide insight. Secondly, many of the papers we discuss do not explicitly mention the preference graph, or use graph-theoretic terminology at all~\citep[for instance][]{shapley_topics_1964,milgrom_rationalizability_1990,monderer_fictitious_1997,monderer_potential_1996,krishna_convergence_1998}. In these cases, our `survey' is counterfactual\footnote{As a consequence, our survey differs in an another way, which is that a `comprehensive' coverage of the topic is impossible. It is likely that the results we mention represent only a part of the full relationship between the preference graph and normal-form games.}, asking how these results \emph{could have} been stated or proved, had the preference graph been explicitly used. Overall, we hope to uncover the preference graph's place in game theory, allowing game theorists to apply it to the next generation of problems.


In Section~\ref{sec: graphs} we define the preference graph, and discuss its naming, size, sparsity and computational complexity. We then explain how the preference graph arises from a natural combination of two game theory axioms: strategic equivalence \citep{moulin1978strategically} and ordinal equivalence \citep{cruz_ordinal_2000}, in a manner reminiscent of the axiomatic development of Von Neumann-Morgenstern utility \citep{von2007theory} (Section~\ref{sec: axioms}). In doing so, we shed light on an important question: how does the preference graph compare to the very similar---and somewhat better-studied---\emph{best-response graph}? This is the subject of the study of \cite{pangallo_best_2019}, who demonstrate that the length and frequency of best-response cycles approximate the empirical behaviour of a broad class of dynamics. We find, somewhat counter-intuitively, that many dynamics have a stronger theoretical relationship with the preference graph than the best-response graph, including `best-response' dynamics like fictitious play (FP). 

In Section~\ref{sec: structure}, we examine how the preference graph influences game structure. Dominated strategies and pure Nash equilibria are both preference graph properties, which motivate several well-studied classes including \emph{dominance-solvable games} \citep{moulin_dominance_1984}, \emph{congestion games} \citep{rosenthal_class_1973}, \emph{potential games} \citep{monderer_potential_1996}, \emph{supermodular games} \citep{topkis1979equilibrium} and \emph{weakly acyclic games} \citep{young_evolution_1993}. These classes turn out to be united by a common thread: their key properties derive from the preference graph. Potential games, for instance, are a famous class of games which capture many economic examples, including congestion games. In introducing potential games, \cite{monderer_potential_1996} are particularly interested in a further generalisation called \emph{ordinal potential games}. This turns out to be expressed very neatly by the preference graph: a game is ordinal potential if and only if its preference graph is \emph{acyclic} \citep{biggar_graph_2023}. Note that this characterisation explains the existence of PNEs in potential games in a simple way: acyclic graphs have sinks, and sinks of the preference graph are PNEs. Similar stories exist for supermodular games \citep{milgrom_rationalizability_1990}, with the associated ordinal concept of \emph{quasi-supermodular games} being represented by the preference graph (Section~\ref{sec:supermodular}). Weakly acyclic and dominance-solvable games too are characterised by properties of their preference graphs, and their defining properties often admit graphical proofs (Section~\ref{sec: weakly acyclic} and \ref{sec:dominance}).

However, the field where the preference graph appears most prolifically is \emph{game dynamics}, which we explore in Section~\ref{sec: dynamics}. A particularly apt example concerns the history of \emph{fictitious play} (FP) \citep{brown1949some}, the best-studied game dynamic (Section~\ref{sec: FP}). In a now-famous paper, \cite{shapley_topics_1964} gave the first example of a game where FP did not converge to an equilibrium. Most surprisingly, Shapley defined this game not by a payoff matrix, but rather by a specific ordering of the strategies for each player---that is, a \emph{preference graph}---demonstrating that the non-convergence of FP in this game was a property of its graph. In Shapley's game, all paths in the preference graph lead to a single sink equilibrium, which is a cycle of length six (Figure~\ref{fig:shapley}). A number of follow-up works explore Shapley's game, understanding the significant constraints of the preference graph on the behaviour of FP~\citep{monderer_fictitious_1997,berger_browns_2007,berger_two_2007,krishna_convergence_1998}.
A similar game was introduced by \cite{jordan_three_1993}, where, again, all FP paths lead to a cycle of length six. This game is sometimes called \emph{Mismatching Pennies}~\citep{sandholm2010population, kleinberg_beyond_2011}. 
The non-convergence to Nash equilibria in Jordan's game is a key component of the work of \citeauthor{kleinberg_beyond_2011}, who show that the social welfare on this 6-cycle in Jordan's game can be arbitrarily higher than at the Nash equilibrium. Potential, supermodular, weakly acyclic and dominance-solvable games are also closely connected to dynamics, with potential games explicitly motivated by convergence of FP \citep{monderer_potential_1996,monderer_fictitious_1996}. \citeauthor{monderer_fictitious_1996} conjectured that \emph{ordinal} potential games also guarantee convergence of FP, a fact eventually proved by \cite{berger_browns_2007,berger_two_2007}, establishing that the two key properties of potential games---existence of PNEs and convergence of FP---derive from acyclicity of the preference graph.

In evolutionary game theory \citep{smith1973logic}, the best-studied dynamic is the \emph{replicator dynamic} \citep{taylor_evolutionary_1978}, which we discuss in Section~\ref{sec: replicator}. Despite its different origins, the replicator dynamic---like FP---turns out to also possess a close relationship with the preference graph \citep{papadimitriou_game_2019,biggar_replicator_2023,biggar_attractor_2024}. Any attractor of the replicator dynamic must contain a sink equilibrium, and a subgame is an attractor if and only if it is a sink equilibrium \citep{gaunersdorfer_fictitious_1995,biggar_replicator_2023}. In zero-sum games the unique attractor is characterised by the unique\footnote{The uniqueness of Nash equilibria in zero-sum games has a preference graph analogue: the sink equilibrium is unique in zero-sum games \citep{biggar_graph_2023}.} sink equilibrium~\citep{biggar_attractor_2024}.
This reflects non-trivial connections between equilibria and preference graphs in zero-sum games: the unique Nash equilibrium lies `inside' the unique sink equilibrium, and when the Nash equilibrium is fully mixed the preference graph is strongly connected \citep{biggar_attractor_2024} (Section~\ref{sec: zero-sum}). This structural relationship is mirrored in the behaviour of dynamics: \cite{papadimitriou2016nash} prove that the replicator dynamic is Poincar\'e recurrent in zero-sum games with interior equilibria, meaning that almost all mixed profiles lie on almost-periodic trajectories. The preference graph is a motivation for the proposal of \emph{sink chain components} as a solution concept for game dynamics \citep{papadimitriou_game_2019}, inspired by the Fundamental Theorem of Dynamical Systems \citep{conley_isolated_1978}. Specifically, \citeauthor{papadimitriou_game_2019} propose using a Markov chain on the preference graph to approximate the behaviour of the replicator dynamic. The idea of approximating dynamics by Markov chains on the preference or best-response graphs is a recurring theme in the literature \citep{young_evolution_1993,goemans_sink_2005,fabrikant2008complexity,omidshafiei_-rank_2019}, with roots as far back as \cite{cournot1838recherches}. We discuss this in Section~\ref{sec: markov}. Finally in Section~\ref{sec: other dynamics} we discuss how these findings generalise to broader classes of dynamics. In Section~\ref{sec: conclusions} we conclude and discuss a variety of interesting research questions inspired by the preference graph.

The scope of each theoretical field is defined by the mathematical tools we employ. Concepts like \emph{strong connectedness} and \emph{cycles} are not just ordinal or combinatorial; they are specifically \emph{graph-theoretic}. 
Graph theory, however, is not a standard part of the mathematical repertoire of game theory. The contribution of the preference graph is not just its results, but the graph theory terminology and proofs through which they are achieved. This paper contains many diagrammatic proofs, demonstrating how graph-theoretic reasoning can provide intuitive and succinct arguments (see for instance Theorems~\ref{dominance theorem} and Lemmas~\ref{shapley uniqueness} and \ref{jordan uniquenesss}). Finally, preference graphs categorise games by their qualitative properties. Often, the same stock of preference graphs---such as the $2\times 2$ and $2\times 3$ graphs (Figures~\ref{fig:2x2} and \ref{fig: 2x3}); Shapley's graph (Figure~\ref{fig:shapley}); Jordan's graph (Figure~\ref{fig:jordan}); and the Inner and Outer Diamond graphs (Figure~\ref{fig:inner diamond and BR})---serve as typical examples or counterexamples to hypotheses about games. Beyond the $2\times 2$ case, these graphs are not well-known, but we believe they are highly valuable tools for practising game theorists, and we hope to raise awareness of their importance.

\section{Preliminaries} \label{sec: preliminaries}
In this paper we study normal-form games with finite players and strategy sets \citep{myerson1997game}. This is defined by a collection of $N$ players, strategy sets $S_1,S_2,\dots,S_N$ and a utility function $u : \prod_{i=1}^{N}S_i \to \real^N$, typically represented by payoff matrices. We call a game with $|S_1|= m_1, |S_2| = m_2, \dots, |S_N| = m_N$ an $m_1\times m_2 \times \dots \times m_N$ game. For simplicity, the set $\prod_{i=1}^{N}S_i$ of tuples of strategies by $Z$. One such tuple (an element of $Z$) is a \emph{strategy profile}. We use the notation $p_{-i}$ to denote an assignment of strategies to all players other than $i$. 
This allows us to write a strategy profile $p$ as a combination of a strategy $s\in S_i$ for player $i$ and the remaining strategies $p_{-i}$, which we denote by $(s ; p_{-i})$. A \emph{subgame} is the game resulting from restricting each player to a subset of their strategies.

One strategy (strictly) \emph{dominates} another if it gives a (strictly) higher payoff for every choice of strategies for the other players. One of the simplest concepts of rationality in games is that dominated strategies should not be played~\citep{myerson1997game}. Deleting dominated strategies from the game can lead to new dominated strategies; iterating this process results in a subgame with no such strategies. If only one profile survives this process, the game is called \emph{dominance-solvable}, with the resultant profile the \emph{dominance solution} (necessarily a PNE) \citep{myerson1997game,moulin_dominance_1984}.

A \emph{mixed strategy} is a distribution over a player's pure strategies, and a \emph{mixed profile} is an assignment of a mixed strategy to each player. Where disambiguation is needed, we will refer to strategies as \emph{pure strategies} and profiles as \emph{pure profiles} to distinguish them the mixed cases. For a mixed profile $x$, we write $x^i$ for the distribution over player $i$'s strategies, and $x^i_s$ for the $s$-entry of player $i$'s distribution, where $s\in S_i$. The set of mixed profiles on a game is given by $\prod_{i=1}^N \Delta_{|S_i|}$ where $\Delta_{|S_i|}$ are the simplices in $\real^{|S_i|}$. We denote $\prod_{i=1}^N \Delta_{|S_i|}$ simply by $X$, and call $X$ the \emph{strategy space} of the game. This is the mixed analogue of $Z$. 

The utility function can be naturally extended to mixed profiles, by taking the expectation over strategies. We denote this by $\exptu : X \to \real^N$, defined by
\begin{equation} \label{def: expected utility}
\exptu_i(x) = \sum_{s_1\in S_1} \sum_{s_2\in S_2} 
    \dots\sum_{s_N\in S_N} \left(\prod_{j = 1}^N x^j_{s_j}\right) u_i(s_1,s_2,\dots,s_N)
\end{equation}
A \emph{Nash equilibrium} of a game is a mixed profile $x$ where no player can increase their expected payoff by a unilateral deviation of strategy. More precisely, a Nash equilibrium is a point where for each player $i$ and strategy $s\in S_i$,
\begin{equation} \label{def: nash equilibrium}
\exptu_i(x) \geq \exptu_i(s;x_{-i})
\end{equation}
A Nash equilibrium is \emph{pure} (a PNE) if $x$ is a pure profile.

A (directed) \emph{graph} is a pair $G = (N,A)$, where $N$ is a finite set of \emph{nodes} and $A\subseteq N\times N$ is a finite set of \emph{arcs}. We denote an arc $(x,y)\in A$ by $\arc{x}{y}$. 
A \emph{path} is a sequence $v_1,v_2,\dots,v_n$ of distinct nodes connected by arcs $\arc{v_i}{v_{i+1}}$. If there is also an arc $\arc{v_n}{v_1}$, we call this a \emph{cycle}. A graph with no cycles is called \emph{acyclic}. If there is a path from a node $v$ to a node $w$ we say \emph{$w$ is reachable from $v$}. Reachability defines a preorder on the nodes of a graph, called the \emph{reachability order}. Two nodes are equivalent under this preorder if both are reachable from each other. The equivalence classes of this relation are called the \emph{strongly connected components}. The minimal elements of this order we call the \emph{sink components}. 

 

\section{Preference graphs} \label{sec: graphs}

The preference graph is a graph defined on the strategy profiles of the game. Two profiles are \emph{$i$-comparable} if they differ in the strategy of player $i$ only, and are \emph{comparable} if they are $i$-comparable for some $i$. 

\begin{defn}[Preference graph]
    The \emph{preference graph} is the graph whose nodes are strategy profiles, with an arc between $i$-comparable profiles in the direction of the higher payoff for player $i$.
\end{defn}

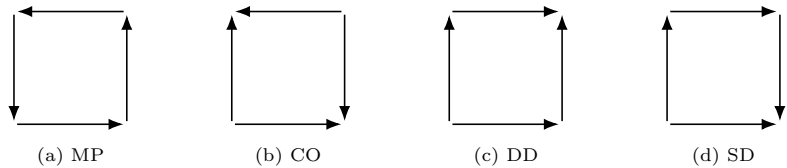
\begin{figure}
    \centering
    \begin{subfigure}{.2\textwidth}
        \centering
        \begin{tikzpicture}
    \node[nd] (A) at (0,0) {};
    \node[nd] (B) at (0,1.5) {};
    \node[nd] (C) at (1.5,0) {};
    \node[nd] (D) at (1.5,1.5) {};
    \draw[arc] (B) to (A);
    \draw[arc] (A) to (C);
    \draw[arc] (C) to (D);
    \draw[arc] (D) to (B);
    \end{tikzpicture}
        \caption{MP}
        \label{fig:MP}
    \end{subfigure}
    \quad
    \begin{subfigure}{.2\textwidth}
        \centering
        \begin{tikzpicture}
    \node[nd] (A) at (0,0) {};
    \node[nd] (B) at (0,1.5) {};
    \node[nd] (C) at (1.5,0) {};
    \node[nd] (D) at (1.5,1.5) {};
    \draw[arc] (A) to (B);
    \draw[arc] (A) to (C);
    \draw[arc] (D) to (B);
    \draw[arc] (D) to (C);
\end{tikzpicture}
        \caption{CO}
        \label{fig:CO}
    \end{subfigure}
    \quad
    \begin{subfigure}{.2\textwidth}
        \centering
        \begin{tikzpicture}
    \node[nd] (A) at (0,0) {};
    \node[nd] (B) at (0,1.5) {};
    \node[nd] (C) at (1.5,0) {};
    \node[nd] (D) at (1.5,1.5) {};
    \draw[arc] (A) to (B);
    \draw[arc] (A) to (C);
    \draw[arc] (C) to (D);
    \draw[arc] (B) to (D);
\end{tikzpicture}
        \caption{DD}
        \label{fig:DD}
    \end{subfigure}
    \quad
    \begin{subfigure}{.2\textwidth}
        \centering
        \begin{tikzpicture}
    \node[nd] (A) at (0,0) {};
    \node[nd] (B) at (0,1.5) {};
    \node[nd] (C) at (1.5,0) {};
    \node[nd] (D) at (1.5,1.5) {};
    \draw[arc] (A) to (B);
    \draw[arc] (B) to (D);
    \draw[arc] (A) to (C);
    \draw[arc] (D) to (C);
\end{tikzpicture}
        \caption{SD}
        \label{fig:SD}
    \end{subfigure}
    \caption{The preference graphs of $2\times 2$ games \citep{biggar_graph_2023}: Matching Pennies (MP,~\ref{fig:MP}), Coordination (CO,~\ref{fig:CO}), Single-Dominance (SD,~\ref{fig:SD}), Double-Dominance (DD,~\ref{fig:DD}).}
    \label{fig:2x2}
\end{figure}

The arcs of the preference graph indicate the \emph{preferred} profile of the (unique) player for whom the two profiles are comparable. Figure~\ref{fig:2x2} depicts the preference graphs of the $2\times 2$ games, of which there are four up to isomorphism. These simple games demonstrate the structure of the preference graph, and form the building blocks for more-complex examples. Sinks of the preference graphs are pure Nash equilibria (PNEs). Three of these games have at least one PNE. Of these three, two are dominance-solvable (Single-Dominance, Figure~\ref{fig:SD} and Double-Dominance, Figure~\ref{fig:DD}). The third (Coordination, Figure~\ref{fig:CO}) has two PNEs, demonstrating the phenomenon of equilibrium selection. The final graph (Matching Pennies, Figure~\ref{fig:MP}) is a 4-cycle, possessing no PNEs, and instead has a fully-mixed equilibrium. 

These graphs give us a qualitative representation of the preference structure of the game. However, sometimes cardinal payoff information may also be needed. There is a canonical way to add this data to the preference graph, which results in the \emph{weighted preference graph} \citep{biggar_graph_2023}.

\begin{defn}[Weighted preference graph] \label{def: weighted preference}
    The \emph{weighted preference graph} is a graph whose nodes are strategy profiles and where there is an arc between $i$-comparable profiles, weighted by the \emph{difference in utility} between those profiles for player $i$.
\end{defn}

If player $i$ receives identical payoff in two $i$-comparable payoffs $p$ and $q$, then there are a pair of arcs $\arc{p}{q}$ and $\arc{q}{p}$ in the preference graph, which we often represent by an  undirected edge $\edge{p}{q}$. This is equivalent to a zero-weighted arc in the weighted preference graph. We call a game `\emph{strict}'~\citep{biggar_graph_2023} if the payoff difference between $i$-comparable profiles is never zero. Games are `generically'~\citep{fudenberg1991game} strict\footnote{Interestingly, while many different `genericity' assumptions about games appear in the literature, several are equivalent to strictness in this sense, including \cite[Section 2.3]{shapley_topics_1964}, \cite[Section 3]{monderer_fictitious_1997} and \cite[Def. 4]{berger_two_2007}.}. For simplicity of presentation, the examples in this paper will usually focus on strict games.

Large preference games can have many arcs, and so can be somewhat cumbersome to draw. In these cases, the \emph{reduced preference graph} can be useful.
\begin{defn} \label{def: reduced preference}
    The \emph{reduced preference graph} of the game is the subgraph of the preference graph where there is an arc $\arc{p}{q}$ between $i$-comparable profiles $p$ and $q$ only if there is no $i$-comparable profile $r$ with arcs $\arc{p}{r}$ and $\arc{r}{q}$.
\end{defn}

The reduced preference graph is constructed by taking the transitive reduction of a player's preferences while holding other players' strategies fixed, hence the name. It has the same reachability order as the preference graph---so the sink equilibria are the same---but much fewer arcs. See Figure~\ref{fig:berger counter}. This has important consequences for computing sink equilibria (see Section~\ref{sec: properties}).

\subsection{Axioms} \label{sec: axioms}

The question of representation of games has its roots in the work of \cite{von2007theory}, who cemented the utility-based approach through a now-famous axiomatic framework. They proved that if players can compare the value of outcomes and \emph{mixtures} of outcomes, then their valuations can be represented by real numbers up to an affine transformation, a concept now called \emph{affine equivalence} or \emph{Von Neumann-Morgenstern equivalence}. However, these assumptions are quite strong, and even \citeauthor{von2007theory} questioned whether a more lenient representation was possible, which was more robust to the payoff assumptions of the model~\citep[Ch. 3]{von2007theory}.

Different approaches have been taken in this direction in game theory, with probably the best-known line of research being \emph{ordinal games} \citep{cruz_ordinal_2000,durieu2008ordinal}. Ordinal games propose that only the \emph{ordering} over outcomes should matter, and defines games up to equivalence of these orderings.  This viewpoint appears very similar to the motivations for preference graphs. There turns out to be a simple and intuitive relationship between ordinal games and preference graphs: \emph{preference equivalence}, the defining equivalence of the preference graphs, derives from a combination of ordinal equivalence and another important game-theoretic equivalence: \emph{strategic equivalence}.

Originating in \cite{moulin1978strategically}, strategic equivalence is a fundamental part of modern game theory. Many game-theoretic concepts---including Nash equilibria, correlated equilibria, fictitious play, the replicator dynamics, and so on---are `invariants' of strategic equivalence, meaning that it is the `strategic' structure of the game which gives rise to these phenomena. We call these \emph{strategic} properties of the game. In the same way, we call invariants of preference equivalence \textbf{preference properties} of the game. Pure NEs, dominance of strategies and sink equilibria are all preference properties. Strategic equivalence differs from affine equivalence by considering only the \emph{relative payoff difference} between \emph{comparable} profiles, on the basis that these profiles are the only ones which can be unilaterally chosen between. This difference is analogous to that between preference graphs and ordinal games: ordinal games model the preferences over \emph{all profiles}, not only comparable ones. By presenting the definitions appropriately, we can make this relationship clear.
\begin{figure}
    \centering
    \begin{tikzpicture}
    \node (A) at (0,4) {\large\textbf{Affine}};
    \node (Str) at (4,2.5) {\large\textbf{Strategic}};
    \node (O) at (-4,2.5) {\large\textbf{Ordinal}};
    \node (P) at (0,1) {\large\textbf{Preference}};
    \draw[->] (A) to (Str);
    \draw[->] (A) to (O);
    \draw[->] (O) to (P);
    \draw[->] (Str) to (P);
    \end{tikzpicture}
    \caption{The relationship between affine, ordinal, strategic and preference equivalence, from \cite{biggarthesis}. The arrows denote that one is a sub-relation of the other.}
    \label{fig:equivalence classes}
\end{figure}
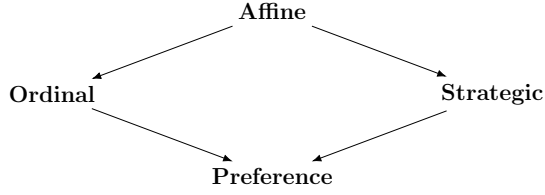
\begin{defn}
    Let $u$ and $v$ be games with the same set of players and strategies. Then $u$ and $v$ are\footnote{The definition of affine equivalence originates in \cite{von2007theory}; strategic equivalence in \cite{moulin1978strategically}; ordinal equivalence in \cite{durieu2008ordinal} and preference equivalence in \cite{biggar_graph_2023}.}:

\begin{tabular}{l l}
     \emph{Affine-equivalent}\footnote{Affine equivalence is typically defined by $u$ and $v$ differing by an affine function. This definition is equivalent, through the famous result of \cite{von2007theory}.} if & $u_p(z)\geq u_p(w) \Leftrightarrow v_p(z) \geq v_p(w)$\\ & for all \textbf{mixed profiles} $z,w \in \Delta$\\
    \emph{Strategic-equivalent} if & $u_p(x;z_{-p}) \geq u_p(y;z_{-p}) \Leftrightarrow v_p(x; z_{-p}) \geq v_p(y; z_{-p})$\\ & for all \textbf{mixed strategies} $x,y\in \Delta_p$ and $z_{-p}\in \Delta_{-p}$.\\
    \emph{Ordinal-equivalent} if & $u_p(a)\geq u_p(b) \Leftrightarrow v_p(a) \geq v_p(b)$ \\ & for all \textbf{pure profiles} $a,b \in Z$. \\
    \emph{Preference-equivalent} if & $u_p(s; r_{-p}) \geq u_p(t; r_{-p}) \Leftrightarrow v_p(s; r_{-p}) \geq v_p(t; r_{-p})$\\ & for all \textbf{pure strategies} $s,t\in S_p$ and $r_{-p}\in Z_{-p}$. 
\end{tabular}
for all players $p$.
\end{defn}
In other words, ordinal equivalence restricts affine equivalence by considering only pure profiles, while strategic equivalence restricts affine by considering only comparable profiles. Preference equivalence combines both. 
\begin{rmk}
    These equivalences are highly important for reasoning about different solution concepts which may be invariants of different classes. (Relative) social welfare, for instance, is an affine invariant, while Nash equilibria are a strategic invariant. These two concepts, somewhat counter-intuitively, are essentially independent---given any game and a specific profile $p$, there is a strategic-equivalent game (and hence with identical Nash equilibria) where that profile is socially optimal\footnote{The same construction works also for Pareto-optimality, which is an ordinal invariant.} \citep{biggarthesis}. The famous Prisoner's Dilemma~\cite{myerson1997game} is an example of this phenomenon. The proof is as follows: choose two players $i$ and $j$, and add some large constant $M$ (larger than the social welfare in any profile) to their payoff $u_i$ ($u_j$) whenever the other players play $p_{-i}$ ($p_{-j}$). This is strategic-equivalent---the relative payoff for any player and unilateral deviation are unchanged. However, $p$ has social welfare at least $2M$, while all other profiles have social welfare less than $2M$.
\end{rmk}

The following theorem relates strategic and preference equivalence to the weighted and unweighted preference graphs.

\begin{thm}[\cite{biggar_graph_2023,candogan_flows_2011}]
    Let $u$ and $v$ be games with the same set of players and strategies. Then $u$ and $v$ are preference-equivalent if and only if their preference graphs are equal, and strategic-equivalent if and only if their weighted preference graphs are equal up to rescaling by a positive constant for each player.
\end{thm}

In other words, the weighted preference graph is a natural representation of strategic equivalence and the preference graph is its `ordinalisation'. 
The fact that the weighted preference graph is a representation of games up to strategic equivalence is quite useful, and is the basis of techniques like the \emph{harmonic decomposition of games} \citep{candogan_flows_2011,hwang_strategic_2020}. 

There is another important equivalence in game theory, which is rarely explicitly discussed \citep{shapley_topics_1964}: two games are equivalent if they are the same up to reordering of players and strategies. Unlike payoff matrices, the graph representations of preference and strategic games automatically capture this equivalence!

\begin{thm}[\cite{biggar_graph_2023}] \label{reconstruction theorem}
    Given a graph $G$ with unlabelled nodes, we can---in linear time---determine if $G$ is the preference graph of a game. If it is, we obtain a labelling of the nodes by strategy profiles which is unique up to renaming of strategies and players.
\end{thm}
Consequently, we can rephrase our previous equivalence in terms of graph isomorphism: two games are preference equivalent \emph{up to renaming of players and strategies} if and only if their preference graphs are \emph{isomorphic}. The same holds for strategic equivalence, \emph{mutatis mutandis}. It is this fact which allows us to draw preference graphs without labelling the nodes by strategy profiles; it is the graph structure---up to isomorphism---which matters. Similarly, while we generally use a `rectangular' style in figures to mimic the matrix structure (Figures~\ref{fig:2x2}, \ref{fig:shapley-square}, \ref{fig: 2x3}, \ref{fig:jordan-square}), we are free to present them in whatever way best displays their structure (see Figures~\ref{fig:shapley-cycle} and \ref{fig:jordan-cycle}). Put differently, the preference graph is truly a graph, not only an ordinal representation of payoff matrices.

The graph structure opens up new research avenues for modelling preferences. An important challenge in game theory occurs when we do not even know a player's preferences over all pairs of alternatives\footnote{One simple approach for modelling indifference of player $i$ between options $a$ and $b$ is to assume \emph{indifference}, that is assuming that $u_i(a) = u_i(b)$. The danger of this approach is that it is assuming far too much; we only know that $i$ prefers \emph{either} $a$ or $b$ ($u_i(a)\geq u_i(b)$ or $u_i(a)\leq u_i(b)$), but we have stated the much stronger fact that $i$ prefers \emph{both} $a$ and $b$ ($u_i(a)\geq u_i(b)$ and $u_i(a)\leq u_i(b)$). The model would actually be more general if we were to choose arbitrarily between $u_i(a)\geq u_i(b)$ and $u_i(a)\leq u_i(b)$.}. Utility functions, by their nature, cannot represent \emph{partially-ordered} preferences of players, and must assume a total order. The preference graph, however, has no such restriction---by omitting arcs, we can easily model such preferences. This could be an important direction for future work (Section~\ref{sec: conclusions}).


\subsection{Sink equilibria}

Pure Nash equilibria (PNEs) are preference properties of a game, represented in the preference graph as sinks \citep{biggar_graph_2023}. This observation motivates the \emph{sink equilibria}.

\begin{defn}
    The \emph{sink equilibria} are the strongly connected components of the preference graph which have no outgoing arcs\footnote{The name `sink equilibria' has been used to describe the sink strongly connected components of both the preference graph and the best-response graph, with the original usage \citep{goemans_sink_2005} applying to the best-response graph. See Section~\ref{sec: better and best} and \ref{sec:nomenclature}.}.
\end{defn}

From the preference graph perspective, sink equilibria are a natural candidate for the solution concept of a game. It is therefore critical to understand how they compare to \emph{Nash equilibria}, the \emph{de facto} standard solution concept for normal-form games \citep{myerson1997game}.

Like (mixed) Nash equilibria, sink equilibria are a generalisation of PNEs and are guaranteed to exist in all games. The proof is simple: every directed graph has sink connected components (by finiteness). Thus while mixed NEs exist by allowing mixed strategies \citep{myerson1997game}, sink equilibria exist by allowing the solution to be a \emph{set} of pure profiles\footnote{Set-wise solution concepts are well-known in game theory, \emph{e.g.} \emph{correlated equilibrium} \citep{aumann1987correlated} and CURB sets \citep{basu_strategy_1991}.}, defined by strict preferences.

At first glance, a set-wise solution concept appears to make less sharp game-theoretic predictions than the Nash equilibrium. However---even putting aside the question of whether the Nash equilibrium is predictive~\citep{milgrom_adaptive_1991,gaunersdorfer_fictitious_1995,hofbauer2011survival,galla_complex_2013}---this intuition is incorrect. Figure~\ref{fig:shapley} shows a well-known $3\times 3$ game where the unique Nash equilibrium is interior, and there is a unique sink equilibrium (a six-cycle on the boundary). While we give a payoff matrix example, these defining properties hold for any game with this preference graph\footnote{This game is called \emph{Shapley's game} \citep{shapley_topics_1964}, and we will discuss it in Section~\ref{sec: FP}.}. If all players were to play the interior equilibrium in this game, then the outcome could be \emph{any} of the nine possible profiles. This includes the three `source' profiles, at which both players wish to deviate. The sink equilibrium excludes these profiles, and so gives a sharper prediction of play than the Nash.

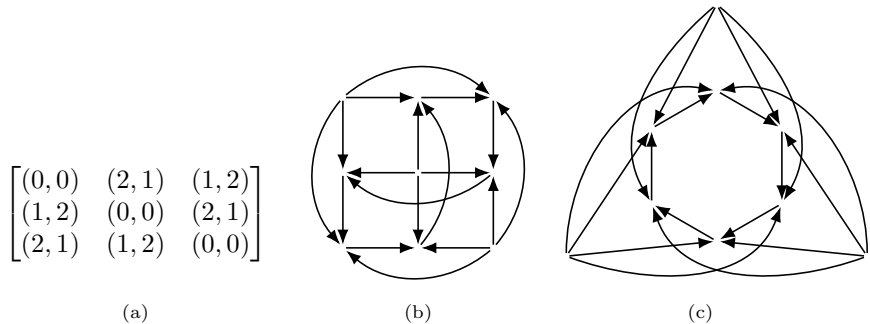
\begin{figure}
    \centering
    \begin{subfigure}{.3\textwidth}
        \centering
        \begin{equation*}
            \begin{bmatrix}
                (0,0) & (2,1) & (1,2) \\
                (1,2) & (0,0) & (2,1) \\
                (2,1) & (1,2) & (0,0)
            \end{bmatrix}
        \end{equation*}
        \caption{}
        \label{fig:shapley payoff}
    \end{subfigure}
    \begin{subfigure}{.3\textwidth}
        \centering
        \begin{tikzpicture}

    \node[nd] (1) at (0,0) {};
    \node[nd] (2) at (1,0) {};
    \node[nd] (3) at (2,0) {};
    \node[nd] (4) at (0,1) {};
    \node[nd] (5) at (1,1) {};
    \node[nd] (6) at (2,1) {};
    \node[nd] (7) at (0,2) {};
    \node[nd] (8) at (1,2) {};
    \node[nd] (9) at (2,2) {};
    
    \draw[arc] (4) to (1);
    \draw[arc] (3) to[in = -40, out = 220] (1);
    \draw[arc] (3) to[out = 50, in = -50] (9);
    \draw[arc] (8) to (9);
    \draw[arc] (5) to (8);
    \draw[arc] (5) to (4);
    
    \draw[arc] (7) to (4);
    \draw[arc] (7) to[in = 130, out = -130] (1);
    \draw[arc] (7) to (8);
    \draw[arc] (7) to[in = 140, out = 40] (9);
    \draw[arc] (1) to (2);
    \draw[arc] (3) to (2);
    \draw[arc] (5) to (2);
    \draw[arc] (2) to[in = -50, out = 50] (8);
    \draw[arc] (3) to (6);
    \draw[arc] (9) to (6);
    \draw[arc] (5) to (6);
    \draw[arc] (6) to[out = 220, in = -40] (4);

 \end{tikzpicture}
        \caption{}
        \label{fig:shapley-square}
    \end{subfigure}
    \begin{subfigure}{.3\textwidth}
        \centering
        \begin{tikzpicture}

    \node[nd] (1) at (0,0) {};
    \node[nd] (2) at (-0.866,0.5) {};
    \node[nd] (3) at (-0.866,1.5) {};
    \node[nd] (5) at (0.866,1.5) {};
    \node[nd] (4) at (0,2) {};
    \node[nd] (6) at (0.866,0.5) {};
    
    \draw[arc] (1) to (2);
    \draw[arc] (2) to (3);
    \draw[arc] (3) to (4);
    \draw[arc] (4) to (5);
    \draw[arc] (5) to (6);
    \draw[arc] (6) to (1);
    
    \node[nd] (7) at (-2,-0.2) {};
    \node[nd] (9) at (2,-0.2) {};
    \node[nd] (8) at (0,3.15) {};

    \draw[arc] (7) to (1);
    \draw[arc] (7) to (3);
    \draw[arc] (7) to[out = -20, in = -110] (6);
    \draw[arc] (7) to[out = 90, in = 170] (4);
    
    \draw[arc] (8) to (5);
    \draw[arc] (8) to (3);
    \draw[arc] (8) to[out = -40, in = 60] (6);
    \draw[arc] (8) to[out = -140, in = 120] (2);
    
    \draw[arc] (9) to (1);
    \draw[arc] (9) to (5);
    \draw[arc] (9) to[out = -160, in = -70] (2);
    \draw[arc] (9) to[out = 90, in = 10] (4);
    
 \end{tikzpicture}
        \caption{}
        \label{fig:shapley-cycle}
    \end{subfigure}
    \caption{Two drawings of Shapley's graph~\citep{shapley_topics_1964}, and a typical payoff matrix representation \citep[\textit{e.g.},][]{krishna_convergence_1998}. Figure~\ref{fig:shapley-square} reflects the structure of the payoff matrix in Figure~\ref{fig:shapley payoff}, while Figure~\ref{fig:shapley-cycle} more clearly shows the connectivity. The graph has three sources which all have arcs into a single sink equilibrium, which is a 6-cycle. This graph appears as the \emph{6-cycle-source graph} in \cite{biggar_graph_2023}.}
    \label{fig:shapley}
\end{figure}

Making a prediction in a game should include some reasoning about how the players end up playing the predicted outcome---this is the goal of game dynamics. 
It is here that the sink and Nash equilibria differ most starkly. We will explore this in Section~\ref{sec: dynamics}, but here we will just observe that there is at least one simple dynamic motivating sink equilibrium play: a random walk on the preference graph always arrives at a sink equilibrium (see Section~\ref{sec: markov}). This dynamic corresponds to players unilaterally and myopically altering their strategy to one which is preferred against the current opponent strategies. For Nash equilibria, by contrast, not even simple dynamics can be made to converge to Nash equilibria in general games \citep{hart_uncoupled_2003,milionis_impossibility_2023}. This is unsurprising given that Nash equilibria are intractable to compute \citep{daskalakis_complexity_2009}, while sink equilibria can be computed efficiently by traversing the graph (Lemma~\ref{lem computing sink equilibria}).
However, convergence to the sink equilibrium appears in many cases to extend to more complex and well-studied dynamics, including fictitious play and the replicator, as we explore in Section~\ref{sec: dynamics}.

Nash equilibria are not the only solution concepts for games, and others, particularly the \emph{correlated}~\citep{aumann1987correlated} and \emph{coarse correlated} equilibrium~\citep{hannan1957approximation} have more appealing computability~\citep{papadimitriou_computing_2008} and dynamic properties~\citep{cesa2006prediction}. The general relationship between these solution concepts and the sink equilibrium is not obvious, and is an important problem for future work (see Section~\ref{sec: conclusions}). For now we will mention only that these solution concepts do not necessarily avoid the problem above, because they are a \emph{superset} of the set of Nash equilibria. For instance, some correlated equilibria in Shapley's game (Figure~\ref{fig:shapley}) contain all three of the source profiles.

\subsection{Symmetric zero-sum games} \label{sec: symmetric zero-sum}

Symmetric zero-sum games are win-lose games where the rules are symmetric for both players. A prototypical example is Rock-Paper-Scissors. Because of the symmetry in the rules, we often restrict our attention to the symmetric profiles, in this case (Rock, Rock), (Scissors, Scissors) and (Paper, Paper). When we restrict to this `symmetric' space, we end up with a different form for the preference graph \citep{biggar_attractor_2024}. Because each profile in consideration is a symmetric pair $(x,x)$ of strategies, we can express each profile $(x,x)$ by the strategy $x$. In this setup, a `profile' is a 1-tuple of strategies, that is, just a strategy. Two `profiles' are comparable if they differ in the strategy of one player, are because these profiles are 1-tuples, all profiles are comparable. Hence, in Rock-Paper-Scissors, the nodes are the strategies Rock, Paper and Scissors, with the arc representing the preferred option in each pair, resulting in Figure~\ref{fig: RPS 1}. Note that, while this form of the graph is quite different, the difference stems from the differences in the definition of `comparable profiles'.

Many dynamics, including the replicator, are restricted to symmetric profiles in a symmetric zero-sum game \citep{smith1973logic}. In such games, this symmetrised graph is the appropriate skeleton of the dynamics \citep{biggar_attractor_2024}.

\subsection{Nomenclature} \label{sec:nomenclature}

The preference graph has gone by many other names, including the \emph{game graph} \citep{candogan_flows_2011}, the \emph{strategy profile graph}~\citep{goemans_sink_2005}, \emph{state graph}~\citep{christodoulou_convergence_2006}, the \emph{Nash dynamics}~\citep{fabrikant_complexity_2004,mirrokni_complexity_2009}, the \emph{better-response graph}~\citep{papadimitriou_game_2019,hakim2024swim} or simply the \emph{response graph} \citep{biggar_graph_2023}. Of these, the last two are probably the most common. While `better-response graph' is useful for contrasting the better- and best-response graphs, it is grammatically awkward and doesn't reflect the simple economic motivation behind the concept of `preference'. `Response graphs', on the other hand, suggests a dynamic concept representing a particular behaviour from players in a repeated game. While the preference graph is useful for analysing game dynamics, it is not \emph{itself} a dynamic---it expresses player's preferences, not the actual sequence of strategies they will play. Overall we think that `preference graphs' is a simple, intuitive name which reflects the concept we are aiming to model.

Sink equilibria, too, have gone by different names, including \emph{sink components}~\citep{biggar_graph_2023} and \emph{Markov-Conley chains}~\citep{papadimitriou_game_2019,omidshafiei_-rank_2019}. Even the original usage of `sink equilibria' \citep{goemans_sink_2005} applied to the \emph{best-response graph} rather than the preference graph. While not an `equilibrium' in the pointwise sense, we think that the name `sink equilibria' captures both their graphical and game-theoretic nature.

\subsection{Basic properties} \label{sec: properties}

In this section we outline some of the simple properties of preference graph. First, we establish which graphs are actually the preference graphs of some game.
\begin{lem} \label{lem: preference graphs as graphs}
    Let $S_1,\dots,S_n$ by finite sets, with $(s_1,\dots,s_n)$ an element of the product $S_1\times\dots\times S_n$. A graph is the preference graph of a game if and only if each node can be labelled by an element $S_1\times\dots\times S_n$ such that that fixing any $s_{-i}$ defines an induced subgraph which is a total order (equivalently, an acyclic orientation of a complete graph).
\end{lem}
This follows from the proof of Theorem~3.2 of \cite{biggar_graph_2023}. In other words, preference graphs are any orientation of a product graph which is a total order along each `row'. This property is straightforward to prove---and to compute, see Theorem~\ref{reconstruction theorem}---but is quite useful, because it allows us to construct preference graphs directly (rather than through payoff matrices) with confidence that they correspond to a game. Because the utility function only influences the orientation of the arcs, we can establish the size and sparsity of the preference graph easily from the number of players and strategies.

\begin{lem} \label{lem: number of nodes and arcs}
    The preference graph of a $m_1\times m_2 \times \dots \times m_n$ game has $\prod_{i=1}^n m_i$ nodes and $(\sum_{i=1}^n (m_i - 1))\prod_{i=1}^n m_i$ arcs.
\end{lem}

This has an immediate computational consequence: the preference graph of a large game is quite sparse, especially if there are many players. 

\begin{lem} \label{lem computing sink equilibria}
    Sink equilibria can be computed in time linear in the number of arcs and nodes.
\end{lem}

The idea of this proof is very simple: simply scan the payoff matrices and compare all comparable pairs of payoffs for each player to obtain the orientation of each arc\footnote{In fact, we can improve on this somewhat using the reduced preference graph (Definition~\ref{def: reduced preference}), as in \cite{hakim2024swim}. For an $m^n$ game represented by $n$ tensors of size $m^n$, the preference graph has $n(m-1)m^n$ arcs, but the reduced preference graph has only $O(nm^n)$. By sorting each row of the tensors we can construct the reduced preference graph in $O(nm^n\log m)$ operations and then search that, giving an overall complexity of $O(m^nn\log m)$. Unfortunately, computing sink equilibria in \emph{succinct games} is hard \citep{fabrikant_complexity_2004,fabrikant2008complexity}.}.






\subsection{Better and best responses} \label{sec: better and best}

The preference graph captures `better-responses'---unilateral deviations which lead to improved payoff for the deviating player. This suggests a natural restriction: instead of any payoff-improving deviation, we add arcs between comparable profiles only when the associated deviation is a \emph{best} response. This results in a subgraph of the preference graph called the \emph{best-response graph}.

Like the preference graph, the best-response graph is an ordinal representation of the underlying structure of the game, and it possesses many of the same philosophical motivations. In fact, because the best-response correspondence is a well-studied object in game theory, the best-response graph appears as the more natural choice of ordinal model for games. Additionally, it is also proportionally better-studied \citep{mirrokni_convergence_2004,goemans_sink_2005,christodoulou_convergence_2006}. For instance, the original definition of sink equilibria referred to the sink connected components of the best-response graph \citep{goemans_sink_2005}. Similarly, the best-response graph was recently proposed as combinatorial skeleton for game dynamics \citep{pangallo_best_2019} because of its close empirical correspondence with dynamic behaviour in randomly-generated games.

It is our finding, however, that this intuition is not correct, and the preference graph is the better model. This is a subtle finding, because in many (small) games these two objects are extremely similar. For instance, in $2\times 2$ games (Figure~\ref{fig:2x2}), the two graphs are identical. Further, the best-response graph is a subgraph of the preference graph, and so its general structure is often fairly analogous. If the preference graph is a good predictor of dynamics in such cases, we would expect the best-response to often be likewise. However, these graphs do have differences, and where they differ we find that it is the preference graph which offers the greater insight into game structure and dynamics. In this section we will demonstrate some key examples where these two graphs have quite distinct properties, and give an overview of which results of the paper apply to the preference graph only.

It is worth noting, first, that because the best-response graph is an induced subgraph of the preference graph, it can always be derived from the preference graph, but not vice versa. So we do not lose the data of best-responses by studying the preference graph.

The most obvious structural difference between preference and best-responses is that preferences are \emph{closed under taking subgames}. In other words, the preference graph of a subgame is an induced subgraph of the preference graph of the game. The same is not true of the best-response graph, where induced subgraphs of subgames might not contain any arcs at all. Put differently, preferences are a \emph{local} property while best-responses are a \emph{global} property. Closure under subgames allows for inductive reasoning, leading to results like Theorems~\ref{dominance theorem} \citep{biggar_graph_2023} and \ref{QSM theorem} \citep{milgrom_rationalizability_1990} and the characterisations of ordinal potential and quasi-supermodular games.

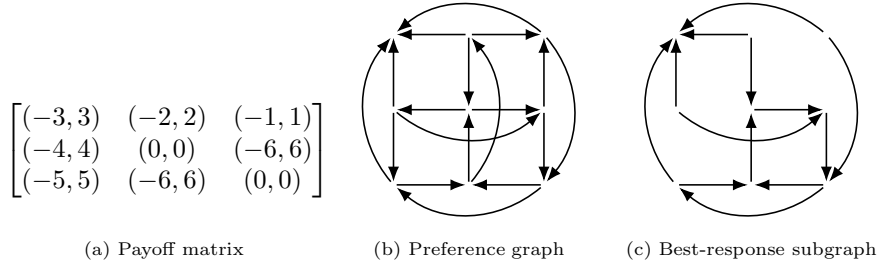
\begin{figure}
    \centering
    \begin{subfigure}{.35\textwidth}
        \centering
        \begin{equation*}
            \begin{bmatrix}
                (-3,3) & (-2,2) & (-1,1) \\
                (-4,4) & (0,0) & (-6,6) \\
                (-5,5) & (-6,6) & (0,0)
            \end{bmatrix}
        \end{equation*}
        \caption{Payoff matrix}
        \label{fig:ID generator}
    \end{subfigure}
    \begin{subfigure}{.3\textwidth}
        \centering
        \begin{tikzpicture}

    \node[nd] (1) at (0,0) {};
    \node[nd] (2) at (1,0) {};
    \node[nd] (3) at (2,0) {};
    \node[nd] (4) at (0,1) {};
    \node[nd] (5) at (1,1) {};
    \node[nd] (6) at (2,1) {};
    \node[nd] (7) at (0,2) {};
    \node[nd] (8) at (1,2) {};
    \node[nd] (9) at (2,2) {};
    
    \draw[arc] (4) to (1);
    \draw[arc] (3) to[in = -40, out = 220] (1);
    \draw[arc] (9) to[in = 50, out = -50] (3);
    \draw[arc] (8) to (9);
    \draw[arc] (8) to (5);
    \draw[arc] (5) to (4);
    
    \draw[arc] (4) to (7);
    \draw[arc] (1) to[out = 130, in = -130] (7);
    \draw[arc] (8) to (7);
    \draw[arc] (9) to[out = 140, in = 40] (7);
    \draw[arc] (1) to (2);
    \draw[arc] (3) to (2);
    \draw[arc] (2) to (5);
    \draw[arc] (2) to[in = -50, out = 50] (8);
    \draw[arc] (6) to (3);
    \draw[arc] (6) to (9);
    \draw[arc] (5) to (6);
    \draw[arc] (4) to[in = 220, out = -40] (6);

 \end{tikzpicture}
        \caption{Preference graph}
        \label{fig:ID}
    \end{subfigure}
    \begin{subfigure}{.3\textwidth}
        \centering
        \begin{tikzpicture}

    \node[nd] (1) at (0,0) {};
    \node[nd] (2) at (1,0) {};
    \node[nd] (3) at (2,0) {};
    \node[nd] (4) at (0,1) {};
    \node[nd] (5) at (1,1) {};
    \node[nd] (6) at (2,1) {};
    \node[nd] (7) at (0,2) {};
    \node[nd] (8) at (1,2) {};
    \node[nd] (9) at (2,2) {};
    
    \draw[arc,opacity=0] (3) to[in = -40, out = 220] (1);
    \draw[arc] (9) to[in = 50, out = -50] (3);
    \draw[arc] (8) to (5);
    
    \draw[arc] (4) to (7);
    \draw[arc] (1) to[out = 130, in = -130] (7);
    \draw[arc] (8) to (7);
    \draw[arc] (9) to[out = 140, in = 40] (7);
    \draw[arc] (1) to (2);
    \draw[arc] (3) to (2);
    \draw[arc] (2) to (5);
    \draw[arc] (6) to (3);
    \draw[arc] (5) to (6);
    \draw[arc] (4) to[in = 220, out = -40] (6);

 \end{tikzpicture}
        \caption{Best-response subgraph}
        \label{fig:ID best response}
    \end{subfigure}
    \caption{The Inner Diamond graph \citep{biggar_graph_2023}, its best-response subgraph and a zero-sum payoff matrix which generates it. The preference graph has a single sink equilibrium, which is a pure Nash equilibrium. The best-response subgraph has an additional sink equilibrium, which is a 4-cycle.}
    \label{fig:inner diamond and BR}
\end{figure}

(Singleton) sinks of the preference graph and best-response graph always coincide, and are exactly the (strict) pure Nash equilibria in both cases. This motivates the definition of sink equilibria in both graphs \citep{goemans_sink_2005}. However, the sink equilibria \emph{can} be different between these two structures when they contain more than one profile. More than that---the best-response graph can have \emph{more} sink equilibria than the preference graph\footnote{The best-response graph cannot have \emph{fewer} sink equilibria than the preference graph, because any sink equilibrium of the preference graph also cannot have arcs out of it in the best-response subgraph.}. In Figure~\ref{fig:inner diamond and BR} we show a preference graph \citep[the Inner Diamond graph,][]{biggar_graph_2023} and its best-response subgraph. This graph has a single sink equilibrium, which is a sink. Its best-response subgraph has an additional sink equilibrium, which is a 4-cycle. In a related work, \cite{amiet_when_2021} showed that in games with probabilistic payoffs, a random walk on the preference graph is highly likely to converge to a PNE when one exists, while a walk on the best-response graph is likely to be trapped in a cycle. Preference graphs can characterise important classes of games---for instance, a game is \emph{ordinal potential} \citep{monderer_fictitious_1996} if and only if its \emph{preference graph} is acyclic (Section~\ref{sec: potential}). Likewise, a key property of \emph{zero-sum games} is that the Nash equilibrium is (generically) unique. This result has a preference graph analogue: the sink equilibrium is unique in zero-sum games \citep{biggar_graph_2023}. The same is not true for the best-response graph: the Inner Diamond graph (Figure~\ref{fig:inner diamond and BR}) can be generated by a zero-sum game (Figure~\ref{fig:ID generator}), but its best-response graph has two sink equilibria.

\begin{figure}
    \centering
    \begin{subfigure}{.3\textwidth}
        \centering
        \begin{equation*}
            \begin{bmatrix}
                (0,0) & (-4,4) \\
                (-1,1) & (-3,3) \\
                (-2,2) & (-1,1)
            \end{bmatrix}
        \end{equation*}
        \caption{Payoff matrix}
        \label{fig:2x3 MP generator}
    \end{subfigure}
    \begin{subfigure}{.3\textwidth}
        \centering
        \begin{tikzpicture}

    \node[nd] (1) at (0,0) {};
    \node[nd] (2) at (1,0) {};
    \node[nd] (3) at (0,1) {};
    \node[nd] (4) at (1,1) {};
    \node[nd] (5) at (0,2) {};
    \node[nd] (6) at (1,2) {};
    
    \draw[arc] (1) to (3);
    \draw[arc] (1) to[out = 130, in = 230] (5);
    \draw[arc] (6) to[out = -50, in = 50] (2);
    \draw[arc] (3) to (5);
    \draw[arc] (4) to (2);
    \draw[arc] (6) to (4);
    
    \draw[arc] (5) to (6);
    \draw[arc] (3) to (4);
    \draw[arc] (2) to (1);
 \end{tikzpicture}
        \caption{Preference graph}
        \label{fig:2x3MP comp}
    \end{subfigure}
    \begin{subfigure}{.3\textwidth}
        \centering
        \begin{tikzpicture}

    \node[nd] (1) at (0,0) {};
    \node[nd] (2) at (1,0) {};
    \node[nd] (3) at (0,1) {};
    \node[nd] (4) at (1,1) {};
    \node[nd] (5) at (0,2) {};
    \node[nd] (6) at (1,2) {};
    
    \draw[arc] (1) to[out = 130, in = 230] (5);
    \draw[arc] (6) to[out = -50, in = 50] (2);
    \draw[arc] (3) to (5);
    \draw[arc] (4) to (2);
    
    \draw[arc] (5) to (6);
    \draw[arc] (3) to (4);
    \draw[arc] (2) to (1);
 \end{tikzpicture}
        \caption{Best-response subgraph}
        \label{fig:2x3MP best response}
    \end{subfigure}
    \caption{The $2\times 3$ MP graph \citep{biggar_graph_2023}, its best-response subgraph and a zero-sum payoff matrix which generates it. The preference graph is strongly connected, but the sink equilibrium of the best-response subgraph is a 4-cycle, which does not contain the support of the Nash equilibrium $(0,1/3,2/3), (2/3,1/3)$. See also Figure~\ref{fig:2x3 MP}.}
    \label{fig:2x3 MP and BR}
\end{figure}
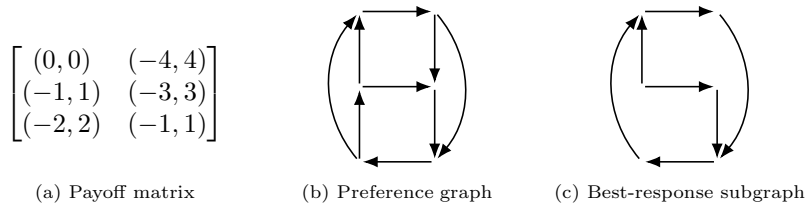
Intuitively, computing best responses is more difficult than better responses---best responses require comparing to all other strategies, while better responses require only comparing pairs of strategies~\citep{sandholm2010population}. In one sense, the difference between and best responses is analogous to the difference between gradient descent and stochastic gradient descent. This difference is reflected in the behaviour of simple dynamics like \emph{fictitious play} (FP) (Section~\ref{sec: FP}). FP is known to converge to equilibria in zero-sum games, and so in the game of Figure~\ref{fig:ID generator} we will converge to the sink of the preference graph, and not the surplus sink equilibrium of its best-response subgraph. In fact, in zero-sum games, it is known that the support of the unique Nash equilibrium must be contained within the sink equilibrium of the preference graph \citep{biggar_attractor_2024}. A simpler example is shown in Figure~\ref{fig:2x3 MP and BR}, where we show a payoff matrix which generates the $2\times 3$ MP graph, the unique preference graph of a $2\times 3$ zero-sum game without dominated strategies \citep{biggar_graph_2023}. This graph is strongly connected, but its best-response subgraph is not. The sink equilibrium of the best-response graph does not contain the support of the unique Nash equilibrium, which is $(0,1/3,2/3), (2/3,1/3)$, and this is where FP will converge \citep{robinson_iterative_1951}. A similar phenomenon holds in ordinal potential games. Acyclicity of the preference graph is sufficient to guarantee convergence of FP in these games \citep{berger_browns_2007,berger_two_2007}. Acyclicity of the best-response graph, by contrast, is not sufficient: \cite{foster_nonconvergence_1998} demonstrate a best-response-acyclic game where FP robustly follows a cycle in the preference graph, where no step is a best response. In other dynamics, like the replicator dynamic, the sink equilibria define the structure of the attractors \citep{biggar_replicator_2023,biggar_attractor_2024}.

The results for FP are especially surprising, because FP is a `best-response' dynamic! At each step, players choose a best response to the empirical distribution of opponent play. However, as these examples show, the overall behaviour is dictated by preference structure rather than best-response structure.


\section{Game classes and structural properties} \label{sec: structure}

The preference graph is a skeleton of the game. That skeleton provides sufficient data to characterise two important game-theoretic objects: dominated strategies and pure Nash equilibria (PNEs). In this section we explore how the graph is often sufficient to prove results on these two concepts. PNEs and dominance are also motivators of many classes of games, including (ordinal) potential games, (quasi-)supermodular games and weakly acyclic games. It turns out that these well-studied game classes can be characterised by the structure of their preference graphs.

\subsection{Potential games} \label{sec: potential}

When they exist, pure Nash equilibria have more theoretically-appealing properties than mixed Nash equilibria. For instance, they are stable under many natural dynamics~\citep{vlatakis-gkaragkounis_no-regret_2020} and don't require players to mix strategies. Their significant drawback is that they do not exist in all games. This observation motivates the study of games where PNE are guaranteed to exist. \citeauthor{rosenthal_class_1973} introduced \emph{congestion games} as one such model. This was subsequently generalised by \citeauthor{monderer_potential_1996}'s \emph{potential games}, a class which is important and now well-established in game theory.

\begin{defn}[Potential games, \cite{monderer_potential_1996}]
    A game $u$ is a \emph{(weighted) potential game} if there exists a function $P: Z \to \real$ and a positive vector $w\in \real^N_+$ such that, for any profile $p\in Z$, player $i$ and strategies $s,t\in S_i$,
    \begin{equation}
        u_i(s; p_{-i}) - u_i(t;p_{-i}) = w_i(P(s; p_{-i}) - P(t;p_{-i}))
    \end{equation}
\end{defn}

\cite{monderer_potential_1996} motivate this class by two properties: 1) PNEs are guaranteed to exist, and 2) a natural learning rule, \emph{fictitious play}, converges to them \citep[proved in][]{monderer_fictitious_1996}. We will discuss fictitious play in the next section. While (weighted) potential games are a cardinal concept, the authors observe that an ordinal generalisation---\emph{ordinal potential games}---are actually sufficient to guarantee the existence of PNEs. 

\begin{defn}[Ordinal potential games, \cite{monderer_potential_1996}]
    A game $u$ is a \emph{ordinal potential game} if there exists a function $P: Z \to \real$ such that, for any profile $p\in Z$, player $i$ and strategies $s,t\in S_i$,
    \begin{equation}
        u_i(s; p_{-i}) - u_i(t;p_{-i}) > 0 \Leftrightarrow P(s; p_{-i}) - P(t;p_{-i}) > 0
    \end{equation}
\end{defn}
\cite{monderer_potential_1996} show that ordinal potential games satisfy the `finite improvement property', a concept which, in our terminology, means that \emph{every walk in the preference graph is finite}. The following characterisation follows in a straightforward way \citep{biggar_graph_2023}:
\begin{thm}
    A game is ordinal potential if and only if its preference graph is acyclic.
\end{thm}


We can now see why ordinal potential games have PNEs: every acyclic graph has a sink, and sinks of the preference graph are PNEs. In short, this class of games can be understood through the graph concept of acyclicity. Interestingly, \citeauthor{monderer_fictitious_1996} conjectured that ordinal potential games are also sufficient to guarantee FP convergence, a fact which was eventually proved by \cite{berger_two_2007,berger_browns_2007} (see Section~\ref{sec: FP}). Ultimately, both of the motivating properties of potential games derive from the preference graph.

\subsection{Supermodular games} \label{sec:supermodular}

A strangely similar story occurs with \emph{supermodular games} \citep{topkis1979equilibrium}, another class motivated by guaranteed existence of PNEs. Supermodular games are defined on games with ordered strategy sets, that is, each player $i$'s strategies have an associated ordering $\leq_i$, Supermodular games are those where these orderings are \emph{complementary}, in that the marginal payoffs to player $i$ for choosing a larger (under $\leq_i$) strategy increases as the competitors strategies increase (under their orderings $\leq_j$). For simplicity of presentation, we will define the two-player case (following \cite{berger_two_2007}). A general definition is provided in \cite{milgrom_rationalizability_1990}.
\begin{defn}[(Two-player) supermodular games \citep{berger_two_2007}]
    Let $u$ be a two-player game with \emph{ordered strategy sets}, that is, both $S_1$ and $S_2$ possess orderings $\leq_1$ and $\leq_2$. $u$ is \emph{supermodular} if for any $s \leq_1 s'$ in $S_1$ and $t \leq_2 t'$ in $S_2$ we have both
    \begin{align*}
        u_1(s',t') - u_1(s,t') & > u_1(s',t) - u_1(s,t) \\
        u_2(s',t') - u_2(s',t) &> u_2(s,t') - u_2(s,t)
    \end{align*}
\end{defn}

Like potential games, this class turns out to have important economic applications, including modelling macroeconomic activity, oligopoly and bank runs, to name a few~\citep{vives1990nash,milgrom_rationalizability_1990}. \citeauthor{milgrom_rationalizability_1990} proved that this class of games possess PNEs---more specifically, they possess a largest and smallest PNE (under the product order on $\leq_i$) which together span the space of non-dominated strategies (Theorem~\ref{QSM theorem}). In the same paper, the authors demonstrate how this causes convergence of `Bayesian learning' to Nash equilibria. In a final parallel to potential games, the authors focus on an ordinal version of supermodularity (sometimes called \emph{quasi-supermodularity} \citep{berger_two_2007}) and demonstrate it is sufficient to guarantee existence of PNEs.

\begin{defn}[(Two-player) quasi-supermodular games \citep{berger_two_2007}] \label{def: quasi supermodular}
    Let $u$ be a two-player game with ordered strategy sets. $u$ is \emph{quasi-supermodular} if for any $s \leq_1 s'$ in $S_1$ and $t \leq_2 t'$ in $S_2$ we have both
    \begin{align*}
        u_1(s',t) > u_1(s,t) &\implies u_1(s',t') > u_1(s,t') \\
        u_2(s,t') > u_2(s,t) &\implies u_2(s',t') > u_2(s',t)
    \end{align*}
\end{defn}

\begin{thm}[\cite{milgrom_rationalizability_1990}] \label{QSM theorem}
    Let $\Gamma$ be a quasi-supermodular game. For each player $p$, there exist largest and smallest iteratively undominated strategies, $\underbar{s}_p$ and $\bar{s}_p$. Moreover the strategy profiles $(\underbar{s}_p)_p$ and $(\bar{s}_p)_p$ are PNEs.
\end{thm}

Quasi-supermodular games, like ordinal potential games, turn out to be a preference property. For simplicity, we will explain the two-player case.
\begin{defn}
    An ordered two-player game is quasi-supermodular if and only if, for any strategies $s_i ,s_j\in S_1$, $s_i <_1 s_j$, $\arc{(s_i,t_i)}{(s_j, t_i)}$ implies $\arc{(s_i,t_j)}{(s_j, t_j)}$ for $t_i <_2 t_j$ in $S_2$, and vice versa for player 2.
\end{defn}
This is a direct graph-theoretic phrasing of Definition~\ref{def: quasi supermodular}. We can use this perspective to understand the proof of \cite{milgrom_rationalizability_1990} in the two-player case in an intuitive combinatorial way (Figure~\ref{fig:milgrom and roberts}).
\begin{figure}
    \centering
    \begin{tikzpicture}

    \node[nd] (00) at (0,0) {$(x_n,y_1)$};
    \node[nd] (01) at (0,1.5) {$\vdots$};
    \node[nd] (02) at (0,2) {};
    \node[nd] (03) at (0,3) {$(x_1,y_1)$};
    \node[nd] (10) at (1,0) {};
    \node[nd] (11) at (1,1) {};
    \node[nd] (12) at (1,2) {};
    \node[nd] (13) at (1,3) {};
    \node[nd] (20) at (2,0) {$\dots$};
    \node[nd] (21) at (2,1) {};
    \node[nd] (22) at (2,2) {};
    \node[nd] (23) at (2,3) {$\dots$};
    \node[nd] (30) at (3,0) {};
    \node[nd] (31) at (3,1) {};
    \node[nd] (32) at (3,2) {};
    \node[nd] (33) at (3,3) {};
    \node[nd] (40) at (4,0) {$(x_n,y_m)$};
    \node[nd] (41) at (4,1) {};
    \node[nd] (42) at (4,1.5) {$\vdots$};
    \node[nd] (43) at (4,3) {$(x_n,y_1)$};

    \draw[arc] (00) to (20);
    \draw[arc] (20) to (40);
    \draw[arc] (00) to (01);
    \draw[arc] (01) to (03);

    \draw[arc] (43) to (23);
    \draw[arc] (23) to (03);
    \draw[arc] (43) to (42);
    \draw[arc] (42) to (40);

 \end{tikzpicture}
    \caption{Assume all iteratively strictly dominated strategies have been removed. Consider the $y_1$ row. If $\arc{(x_i, y_1)}{(x_j,y_1)}$ for $i<j$, then $\arc{(x_i, y_k)}{(x_j,y_k)}$ for every $y_k$, which is a contradiction as implies $x_j$ dominates $x_i$. Consequently, the $x$ strategies in order in the $y_1$ row. The same argument, the $x$ strategies are in reverse order in the $y_n$ row. The columns follow the same reasoning.}
    \label{fig:milgrom and roberts}
\end{figure}
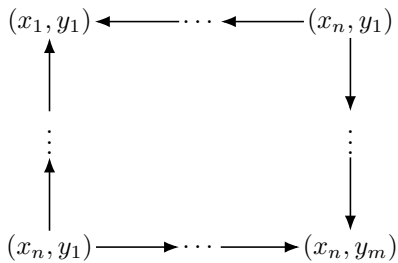

\cite{milgrom_rationalizability_1990} conjectured that \emph{quasi}-supermodularity should be sufficient to guarantee convergence of natural dynamics, like fictitious play. Using graphical ideas, \cite{berger_two_2007} proved a restricted form of this conjecture by relating quasi-supermodular games to ordinal potential games: every two-player quasi-supermodular game where one player has at most three strategies, or both players have at most four, is ordinal potential and hence converges under fictitious play. Unfortunately, Berger demonstrated that this argument does not extend further, because $4\times 5$ quasi-supermodular games may not have acyclic preference graphs\footnote{\cite{berger_two_2007} also proved the conjecture in all two-player quasi-supermodular games with \emph{diminishing returns}. To our knowledge, the conjecture remains unresolved in general quasi-supermodular games.} (Figure~\ref{fig:berger counter}).

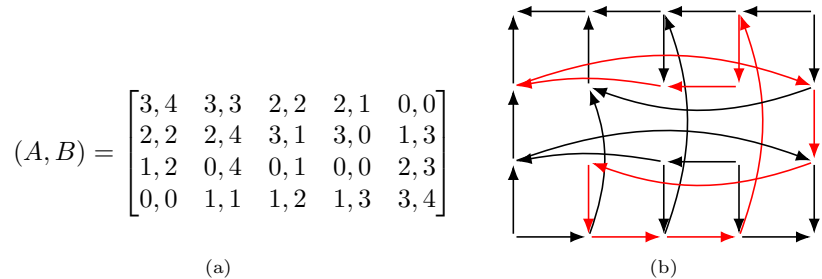
\begin{figure}
    \centering
    \begin{subfigure}{.45\textwidth}
        \centering
        \begin{equation*}
            (A,B) = \begin{bmatrix}
                3,4 & 3,3 & 2,2 & 2,1 & 0,0 \\
                2,2 & 2,4 & 3,1 & 3,0 & 1,3 \\
                1,2 & 0,4 & 0,1 & 0,0 & 2,3 \\
                0,0 & 1,1 & 1,2 & 1,3 & 3,4
            \end{bmatrix}
        \end{equation*}
        \caption{}
        \label{fig:berger matrixl}
    \end{subfigure}
    \quad
    \begin{subfigure}{.45\textwidth}
        \centering
        \begin{tikzpicture}

    \node[nd] (00) at (0,0) {};
    \node[nd] (01) at (0,1) {};
    \node[nd] (02) at (0,2) {};
    \node[nd] (03) at (0,3) {};
    \node[nd] (10) at (1,0) {};
    \node[nd] (11) at (1,1) {};
    \node[nd] (12) at (1,2) {};
    \node[nd] (13) at (1,3) {};
    \node[nd] (20) at (2,0) {};
    \node[nd] (21) at (2,1) {};
    \node[nd] (22) at (2,2) {};
    \node[nd] (23) at (2,3) {};
    \node[nd] (30) at (3,0) {};
    \node[nd] (31) at (3,1) {};
    \node[nd] (32) at (3,2) {};
    \node[nd] (33) at (3,3) {};
    \node[nd] (40) at (4,0) {};
    \node[nd] (41) at (4,1) {};
    \node[nd] (42) at (4,2) {};
    \node[nd] (43) at (4,3) {};

    \draw[arc] (00) to (10);
    \draw[arc,red] (10) to (20);
    \draw[arc,red] (20) to (30);
    \draw[arc] (30) to (40);
    \draw[arc] (00) to (01);
    \draw[arc] (01) to (02);
    \draw[arc] (02) to (03);

    \draw[arc] (43) to (33);
    \draw[arc] (33) to (23);
    \draw[arc] (23) to (13);
    \draw[arc] (13) to (03);
    \draw[arc] (43) to (42);
    \draw[arc,red] (42) to (41);
    \draw[arc] (41) to (40);

    \draw[arc,red] (11) to (10);
    \draw[arc] (10) to [out = 70, in=-70] (12);
    \draw[arc] (12) to (13);

    \draw[arc] (21) to (20);
    \draw[arc] (20) to [out = 70, in=-70] (23);
    \draw[arc] (23) to (22);

    \draw[arc] (31) to (30);
    \draw[arc,red] (30) to [out = 70, in=-70] (33);
    \draw[arc,red] (33) to (32);

    \draw[arc,red] (32) to (22);
    \draw[arc,red] (22) to [out = 170, in= 10] (02);
    \draw[arc,red] (02) to [out = 20, in=160] (42);
    \draw[arc] (42) to [out = -160, in=-20] (12);

    \draw[arc] (31) to (21);
    \draw[arc] (21) to [out = 170, in= 10] (01);
    \draw[arc] (01) to [out = 20, in=160] (41);
    \draw[arc,red] (41) to [out = -160, in=-20] (11);

 \end{tikzpicture}
        \caption{}
        \label{fig:berger graph}
    \end{subfigure}
    \caption{The $4\times 5$ quasi-supermodular game from \cite[pg. 12]{berger_two_2007} (\ref{fig:berger matrixl}) and its (reduced) preference graph (\ref{fig:berger graph}) (Section~\ref{sec: graphs}). Berger demonstrates that such games can have cycles in their preference graphs (highlighted in red), and hence are not ordinal potential.}
    \label{fig:berger counter}
\end{figure}

\subsection{Weakly acyclic games} \label{sec: weakly acyclic}


A third class of games possessing PNEs is the \emph{weakly acyclic games} \citep{young_evolution_1993}. Like ordinal potential games, weakly acyclic games are defined by paths in the preference graph.

\begin{defn}[Weakly acyclic games \citep{young_evolution_1993}]
    A game $u$ is \emph{weakly acyclic}\footnote{\cite{young_evolution_1993} analogously defines \emph{acyclic games} as those whose best-response graph is acyclic.} (WA) if there is a path from every node in the best-response graph to a sink.
\end{defn}

It is straightforward to see that a game is weakly acyclic if and only if every sink equilibrium of the best-response graph is a sink. Notably, unlike the case for potential and supermodular games, Young makes explicit use of graph-theoretic ideas in defining these games. However, like the previous two classes, Young motivates weak acyclicity by the convergence of \emph{adaptive dynamics}, a stochastic procedure modelling best-response play from players with finite memory (see Section~\ref{sec: other dynamics}). Weakly acyclic games are a best-response graph property; analogously, we could define \emph{better-response weakly acyclic} (BRWA) games as those where every sink equilibrium of the \emph{preference graph} is a sink. These games are a strict subset of the weakly acyclic games. As an example, we showed in Section~\ref{sec: better and best} the Inner Diamond game (Figure~\ref{fig:inner diamond and BR}) is BRWA but not WA.

\subsection{Zero-sum games} \label{sec: zero-sum}
Note that, while not characterised solely by their graph properties, \emph{zero-sum games} \citep{von2007theory} also possess non-trivial graph structure which is sufficient to understand convergence of the replicator dynamic \citep{biggar_attractor_2024}. A useful fact for analysing zero-sum games is that they are `reflections' of potential games \citep{candogan_flows_2011,hwang_strategic_2020}. That is, if we negate the payoffs to one player in a zero-sum game, the result is an identical-interest game, and a game is potential if and only if it is strategically equivalent to an identical interest game~\citep{monderer_fictitious_1996,monderer_potential_1996}. The result is that a preference graph is generated by some zero-sum game if and only if it is acyclic after \emph{reflection} (reversing the preference of one player).
\begin{thm}[\cite{biggar_graph_2023}] \label{ZS characterisation}
    A two-player game $(u,v)$ is preference-equivalent to a zero-sum game if and only if $(u,-v)$ has an acyclic preference graph, \emph{i.e.}, $(u,-v)$ is ordinal potential.
\end{thm}
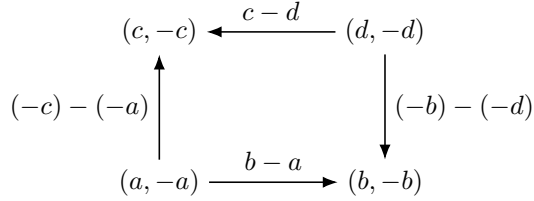
\begin{figure}
    \centering
    \begin{tikzpicture}
    \node (bl) at (0,0) {$(a,-a)$};
    \node (br) at (3,0) {$(b,-b)$};
    \node (tl) at (0,2) {$(c,-c)$};
    \node (tr) at (3,2) {$(d,-d)$};
    \draw[arc] (bl) to node[above] {$b-a$} (br);
    \draw[arc] (tr) to node[above] {$c-d$} (tl);
    \draw[arc] (tr) to node[right] {$(-b)-(-d)$} (br);
    \draw[arc] (bl) to node[left] {$(-c)-(-a)$} (tl);
\end{tikzpicture}
    \caption{Theorem~\ref{ZS characterisation} can be visualised by MP (Figure~\ref{fig:CO}), whose reflection is CO. This diagram from \cite{biggar_graph_2023} shows that zero-sum games cannot have the CO structure, because it induces the impossible strict cycle $a>c>d>b>a$.}
    \label{fig:no CO}
\end{figure}
This can be used to prove that the preference graph of a zero-sum game always has a unique sink equilibrium. The proof is graph-theoretic. The idea is that the existence of multiple sink equilibria imply the existence a $2\times 2$ subgame whose structure is CO (Figure~\ref{fig:CO}). But the reflection of CO is MP, which is a cycle! See Figure~\ref{fig:no CO}. A corollary is that if a zero-sum game has a PNE it must be unique \citep{biggar_graph_2023}.



\subsection{Dominance-solvability and the preference graph} \label{sec:dominance}

Dominated strategies are a fundamental game theory concept, and are a preference property. Unsurprisingly, the preference graph can help us to analyse dominance. In this section we explore some other sufficient conditions for dominance-solvability, and discover how the graph leads to neat proofs.

\cite{shapley_topics_1964} studied the question of when a zero-sum game was dominance-solvable. He gave a simple sufficient condition: if every $2\times 2$ subgame has a PNE, then the game is dominance-solvable. The proof is a clever application of the von Neumann-Morgenstern value of a zero-sum game \citep{von2007theory}.

Recalling the preference graphs of the $2\times 2$ games, we can rephrase Shapley's statement into a graph-theoretic one: if every $2\times 2$ subgraph of a zero-sum game is either CO (Fig.~\ref{fig:CO}), SD (Fig.~\ref{fig:SD}) or DD (Fig.~\ref{fig:DD}) then the game is dominance-solvable. In fact, because we know that CO subgames cannot exist in zero-sum games (Figure~\ref{fig:no CO}), this result is equivalent to the statement that if every $2\times 2$ subgame of a zero-sum game is dominance-solvable (SD or DD), then the game is dominance-solvable.

This intuition leads to a more general result, which holds in all two-player games (not only zero-sum). We reproduce here the proof of \cite{biggar_graph_2023}, to demonstrate the use of the preference graph.

\begin{thm}[\cite{biggar_graph_2023}] \label{dominance theorem}
   If every $2\times 2$ subgame of a two-player game is dominance-solvable, then the game is dominance-solvable.
\end{thm}
\begin{proof}
    For simplicity, we will assume the game is strict\footnote{This assumption changes very little, and a full proof exists in \cite{biggar_graph_2023}.}. If each $2\times 2$ subgame is dominance-solvable, then each is either SD or DD (Figures~\ref{fig:SD} and \ref{fig:DD}), or equivalently none are MP or CO (Figures~\ref{fig:MP} and \ref{fig:CO}). First, assume for contradiction that no dominated strategies exist. Now, fix a strategy $h \in S_2$ for player 2, and let $s_1,s_2,\dots,s_n \in S_1$ be the strategies of player 1, in order of preference given $h$. Because $s_n$ does not dominate $s_1$ (by assumption), there must be another strategy $k\in S_2$ where player 1 prefers $s_1$ over $s_n$ given $k$.
    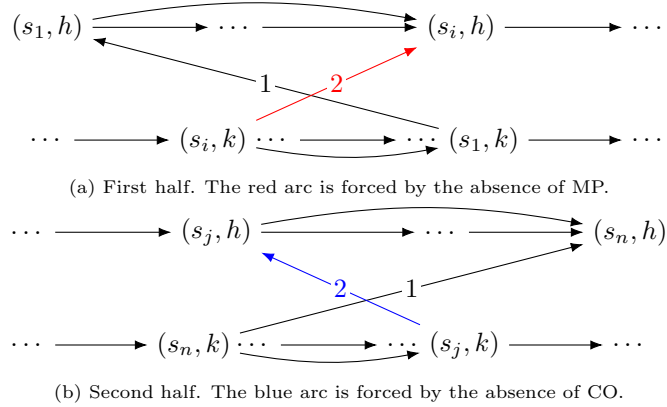
\begin{figure}
    \centering
        \begin{subfigure}{\textwidth}
            \centering
            \begin{tikzpicture}
    
    \node (s1) at (0,1.5) {$(s_1,h)$};
    \node (d1) at (2.5,1.5) {$\dots$};
    \node (d2) at (5.5,1.5) {$(s_i,h)$};
    \node (sn) at (8,1.5) {$\dots$};
    
    \node (s1b) at (0,0) {$\dots$};
    \node (swb) at (2.2,0) {$(s_i,k)$};
    \node (d2b) at (3,0) {$\dots$};
    \node (d3b) at (5,0) {$\dots$};
    \node (sab) at (5.8,0) {$(s_1,k)$};
    \node (snb) at (8,0) {$\dots$};
    
    \draw[->] (s1) to (d1);
    \draw[->] (d1) to (d2);
    \draw[->] (d2) to (sn);
    
    \draw[->] (s1b) to (swb);
    \draw[->] (d2b) to (d3b);
    \draw[->] (sab) to (snb);
    
    \draw[->] (sab) to node[label=above] {1} (s1);
    \draw[->] (s1) to [out = 10, in = 170] (d2);
    \draw[->] (swb) to [out = -10, in = -170] (sab);

    \draw[->,red] (swb) to node[label=above] {2} (d2);
 \end{tikzpicture}
            \caption{First half. The red arc is forced by the absence of MP.}
            \label{fig:dominance 2}
        \end{subfigure}
        \begin{subfigure}{\textwidth}
            \centering
            \begin{tikzpicture}
    
    \node (s1) at (0,1.5) {$\dots$};
    \node (d1) at (2.5,1.5) {$(s_j,h)$};
    \node (d2) at (5.5,1.5) {$\dots$};
    \node (sn) at (8,1.5) {$(s_n,h)$};
    
    \node (s1b) at (0,0) {$\dots$};
    \node (swb) at (2.2,0) {$(s_n,k)$};
    \node (d2b) at (3,0) {$\dots$};
    \node (d3b) at (5,0) {$\dots$};
    \node (sab) at (5.8,0) {$(s_j,k)$};
    \node (snb) at (8,0) {$\dots$};
    
    \draw[->] (s1) to (d1);
    \draw[->] (d1) to (d2);
    \draw[->] (d2) to (sn);
    
    \draw[->] (s1b) to (swb);
    \draw[->] (d2b) to (d3b);
    \draw[->] (sab) to (snb);
    
    \draw[->] (d1) to [out = 10, in = 170] (sn);
    \draw[->] (swb) to [out = -10, in = -170] (sab);

    \draw[->] (swb) to node[label=above] {1} (sn);
    \draw[->,blue] (sab) to node[label=above] {2} (d1);
 \end{tikzpicture}
            \caption{Second half. The blue arc is forced by the absence of CO.}
            \label{fig:dominance 3}
        \end{subfigure}
        \caption{The graphical structure of the proof of Theorem~\ref{dominance theorem}.}
        \label{fig:dominance proof}
    \end{figure}
    
    We will assume there is an arc $\arc{(s_1,k)}{(s_1,h)}$\footnote{The opposite case $\arc{(s_1,h)}{(s_1,k)}$ is symmetric, with the roles of MP and CO reversed.}. The proof has two parts. First, we consider any strategy $s_i$ where $\arc{(s_i,k)}{(s_1,k)}$. The subgame $\{s_i,s_1\}\times \{h,k\}$ cannot be MP, so we conclude there is an arc $\arc{(s_i,k)}{(s_i,h)}$ (Figure~\ref{fig:dominance 2}). In particular, since $s_n$ has an arc $\arc{(s_n,k)}{(s_1,k)}$, we conclude there is an arc $\arc{(s_n,k)}{(s_n,h)}$. This leads us to the second part, where we consider strategies $s_j$ where $\arc{(s_n,k)}{(s_j,k)}$. The subgame $\{s_j,s_n\}\times \{h,k\}$ cannot be CO, so we conclude there is an arc $\arc{(s_j,k)}{(s_j,h)}$ (Figure~\ref{fig:dominance 3}). We find that $h$ dominates $k$, which is a contradiction.
\end{proof}

There are several interesting variants of this question which have been explored, and often have a graph-theoretic flavour. A broad generalisation of Shapley's theorem is the work of \cite{boros2016sufficient}. In it the authors provide a \emph{complete description} of when the structure of the $2\times 2$ subgames of a two-player game is sufficient to guarantee the existence of a PNE. They do this by examining all possible (non-strict) $2\times 2$ preference graphs, and show how their combinations force structure on the overall preference graph. \cite{takahashi2002pure} proved that existence of a PNE in \emph{every} subgame of a two-player game implies the game is weakly acyclic. \cite{fabrikant_structure_2010} generalised this to $n$-player games, proving that the existence of a \emph{unique} PNE in all subgames implies the game is weakly acyclic.

Zero-sum games do not contain CO; Theorem~\ref{dominance theorem} implies that every non-iteratively-dominated strategy must participate in at least one MP. Consequently, the preference graphs of zero-sum games are often highly connected. An interesting structural question is whether the converse is true: high connectedness implies 4-cycles.

\begin{conj}
    Every strongly connected two-player preference graph contains a 4-cycle.
\end{conj}

While dominance of strategies is captured by the preference graph, it is not always reflected in the \emph{connectivity} of the graph. This is somewhat counter-intuitive, and means that sink equilibria can sometimes contain dominated strategies. In Figure \ref{fig:connected dominance} we show a $2\times 3$ game which is strongly connected, so all profiles are contained in the sink equilibrium, yet there is a dominated strategy. Common dynamics like the replicator and fictitious play cause dominated strategies to vanish in the long run, so this shows that the empirical behaviour can be a subset of the sink equilibrium in some cases. See Section~\ref{sec: conclusions}.

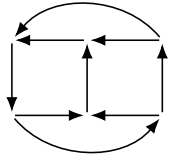
\begin{figure}
    \centering
    \begin{tikzpicture}

    \node[nd] (1) at (0,0) {};
    \node[nd] (2) at (0,1) {};
    \node[nd] (3) at (1,0) {};
    \node[nd] (4) at (1,1) {};
    \node[nd] (5) at (2,0) {};
    \node[nd] (6) at (2,1) {};
    
    \draw[arc] (1) to (3);
    \draw[arc] (1) to[out = -50, in = -130] (5);
    \draw[arc] (6) to[out = 130, in = 50] (2);
    \draw[arc] (5) to (3);
    \draw[arc] (4) to (2);
    \draw[arc] (6) to (4);
    
    \draw[arc] (5) to (6);
    \draw[arc] (3) to (4);
    \draw[arc] (2) to (1);
 \end{tikzpicture}
    \caption{A $2\times 3$ game that is strongly connected, yet contains a dominated strategy; the second column dominates the third.}
    \label{fig:connected dominance}
\end{figure}






\section{Game dynamics} \label{sec: dynamics}

In this section we explore how the preference graph relates to a variety of different well-studied game dynamics.

\subsection{Fictitious play} \label{sec: FP}


\emph{Fictitious play} (FP) is the oldest and best-studied learning algorithm in games. It is based on a simple rule: at each round, play a best response to the opponent's empirical distribution of play over all previous rounds. Formally
\begin{defn}[Two-player discrete-time fictitious play]
    A (two-player) \emph{fictitious play sequence} is a sequence of profiles $p_n = (s_1^n, s_2^n) \in S_1\times S_2$ where at each round $n$, $s_1^n \in \arg\max \sum_{i=1}^{n-1} u_1(p_i)$ and $s_2^n \in \arg\max\sum_{i=1}^{n-1} u_2(p_i)$, respectively.
\end{defn}
While the set $\arg\max \sum_{i=1}^{n-1} u_1(p_i)$ generically consists of a single strategy, it is possible for multiple strategies to be best-responses at some time points. By this definition, any choice from these strategies constitutes an FP sequence. Other definitions prescribe particular tie-breaking rules \cite{miyasawa1961convergence,krishna1992learning}, or more generally allow any mixture of best-responses to be played at tie points. The latter gives rise to the best-response dynamics \citep{hofbauer_evolutionary_2003}, and includes more possible sequences than those we analyse here. With some more setup, FP can be expressed in continuous-time rather than discrete rounds, a model which is sometimes easier to work with \citep{krishna_convergence_1998,berger_fictitious_2005,berger_two_2007}.

This algorithm, first presented by~\cite{brown1949some}, was an early hope for explaining Nash equilibrium~\citep{nash_non-cooperative_1951} behaviour. \citeauthor{robinson_iterative_1951} proved that the empirical distributions over strategies played in FP sequences in zero-sum games converges to the Nash equilibrium. This set off a line of research in game dynamics which continues to the modern day. For fictitious play, proofs of convergence in various classes of games were steadily achieved: $2\times 2$ games~\citep{miyasawa1961convergence}, $2\times 3$ games \citep{monderer_fictitious_1997}, potential games \citep{monderer_potential_1996}, dominance-solvable games~\cite{milgrom_adaptive_1991}, quasi-supermodular games \citep{milgrom_rationalizability_1990} and $2\times N$ games \citep{berger_fictitious_2005}.

Despite these successes, the initial hope that FP would justify Nash equilibrium play in all games was quickly extinguished: \cite{shapley_topics_1964} demonstrated a robust example of a $3\times 3$ game where FP did not converge. It is here that the preference graph enters the story. While often later presented as a particular payoff matrix~\citep{krishna_convergence_1998,jordan_three_1993,gaunersdorfer_fictitious_1995}, Shapley's actual discovery was ``a whole class of order-equivalent games, [which did] not depend on numerical quirks in the payoff matrices." In our terminology, Shapley's example defined a particular preference graph (which we call the \emph{Shapley graph}, shown in Figure~\ref{fig:shapley}), and his proof demonstrated that \emph{any} game with the Shapley graph (which we call a \emph{Shapley game}) cannot converge under FP. In other words, the obstacle to FP convergence is the graph structure.

While nowadays typically defined with player's updating their beliefs simultaneously (as above), \citeauthor{brown1949some}'s original FP described the players as updating their preferences alternately \citep{berger_browns_2007}. As each player always selects a strategy that improves on the current one, this results in every FP sequence being a walk on the preference graph\footnote{But is not necessarily a walk on the best-response graph---recall Section~\ref{sec: better and best}.}. In ordinal potential games, this gives us a simple graph-theoretic proof of convergence: every walk must terminate after finitely many steps at a PNE \citep{berger_browns_2007}. \cite{monderer_fictitious_1997} and \cite{berger_two_2007, berger_browns_2007} studied `improvement principles' for FP, which explain more broadly how the preference graph constrains FP sequences. Under simultaneous FP, it is possible for both players to update their beliefs simultaneously, leading to a transition that is \emph{not} an arc in the preference graph, but \cite{berger_two_2007} proved that such transitions take place on a subgame with the graph structure of DD (Figure~\ref{fig:DD}), ensuring that reachability in the preference graph is consistent with the profiles reachable under FP sequences.

Ordinal potential games are not an isolated example. Surprisingly many of the milestones in the study of FP can be explained by graph structure. 
Dominance-solvable games, for instance, converge under FP and are characterised by their graph structure. As we saw in the previous section, quasi-supermodular games \citep{milgrom_rationalizability_1990}, like potential games, are defined by graph structure, and this can be used to establish when FP converges \citep{berger_two_2007}. Convergence in $2\times 2$ games can be proved by examining the four possible preference graphs: two are dominance-solvable (SD and DD), one is potential (CO), and the remaining graph is the 4-cycle (MP), which is always strategically zero-sum. A similar procedure can be carried out on the $2\times 3$ games, where only three examples exist without dominated strategies (see Figure~\ref{fig: 2x3}). A case-by-case analysis of these graphs is actually the approach used in \citeauthor{monderer_fictitious_1997}'s original proof of FP convergence in $2\times 3$ games.

\begin{figure}
    \centering
    \begin{subfigure}{0.22\textwidth}
    \centering
    \begin{tikzpicture}

    \node[nd] (1) at (0,0) {};
    \node[nd] (2) at (1,0) {};
    \node[nd] (3) at (0,1) {};
    \node[nd] (4) at (1,1) {};
    \node[nd] (5) at (0,2) {};
    \node[nd] (6) at (1,2) {};
    
    \draw[arc] (1) to (3);
    \draw[arc] (1) to[out = 130, in = 230] (5);
    \draw[arc] (6) to[out = -50, in = 50] (2);
    \draw[arc] (3) to (5);
    \draw[arc] (4) to (2);
    \draw[arc] (6) to (4);
    
    \draw[arc] (6) to (5);
    \draw[arc] (3) to (4);
    \draw[arc] (1) to (2);
 \end{tikzpicture}
    \caption{$2\times 3$ CO}
    \label{fig:2x3 CO}
    \end{subfigure}
    \begin{subfigure}{0.22\textwidth}
    \centering
    \begin{tikzpicture}

    \node[nd] (1) at (0,0) {};
    \node[nd] (2) at (1,0) {};
    \node[nd] (3) at (0,1) {};
    \node[nd] (4) at (1,1) {};
    \node[nd] (5) at (0,2) {};
    \node[nd] (6) at (1,2) {};
    
    \draw[arc] (1) to (3);
    \draw[arc] (1) to[out = 130, in = 230] (5);
    \draw[arc] (6) to[out = -50, in = 50] (2);
    \draw[arc] (3) to (5);
    \draw[arc] (4) to (2);
    \draw[arc] (6) to (4);
    
    \draw[arc] (5) to (6);
    \draw[arc] (3) to (4);
    \draw[arc] (2) to (1);
 \end{tikzpicture}
    \caption{$2\times 3$ MP}
    \label{fig:2x3 MP}
    \end{subfigure}
    \begin{subfigure}{0.22\textwidth}
    \centering
    \begin{tikzpicture}

    \node[nd] (1) at (0,0) {};
    \node[nd] (2) at (1,0) {};
    \node[nd] (3) at (0,1) {};
    \node[nd] (4) at (1,1) {};
    \node[nd] (5) at (0,2) {};
    \node[nd] (6) at (1,2) {};
    
    \draw[arc] (1) to (3);
    \draw[arc] (1) to[out = 130, in = 230] (5);
    \draw[arc] (6) to[out = -50, in = 50] (2);
    \draw[arc] (3) to (5);
    \draw[arc] (4) to (2);
    \draw[arc] (6) to (4);
    
    \draw[arc] (6) to (5);
    \draw[arc] (3) to (4);
    \draw[arc] (2) to (1);
 \end{tikzpicture}
    \caption{$2\times 3$ MP-CO}
    \label{fig:2x3 MPCO}
    \end{subfigure}
    \caption{The three preference graphs of strict $2\times 3$ games without dominated strategies \citep{biggar_graph_2023}}
    \label{fig: 2x3}
\end{figure}
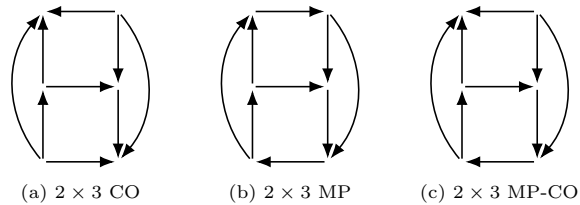


The class which, at first, resists a graph-theoretic explanation of convergence is zero-sum games. However, even this class has non-trivial connections to the preference graph. Consider the $2\times 2$ MP game, whose preference graph is the 4-cycle. Any game with this graph possesses a unique interior equilibrium, and FP always converges to it. It turns out that the 4-cycle is unique in this regard---if FP robustly follows any cycle larger than 4, the empirical distributions of strategies do not converge \citep{krishna_convergence_1998}! 
Revisiting the Shapley graph, we observe that it possesses 3 sources with arcs into a single sink equilibrium, which is a cycle of length 6. Any FP path reaches and follows this cycle \citep{monderer_fictitious_1997,berger_two_2007}, but the empirical distributions of strategies do not converge because the cycle is larger than four. This property actually characterises Shapley's graph: it is the unique smallest preference graph which has neither a sink nor a 4-cycle.

\begin{lem} \label{shapley uniqueness}
    Shapley's graph is the unique smallest two-player game with neither a sink nor 4-cycle.
\end{lem}
\begin{proof}
    Acyclic graphs have sinks, so a graph without a sink must have a cycle. Observe first that if a cycle contains two consecutive deviations $\arc{(a,y)}{(b,y)}$ $\arc{(b,y)}{(c,y)}$ for the same player, then a shorter cycle exists with these arcs replaced by $\arc{(a,y)}{(c,y)}$. In particular, in two players, a minimal cycle must alternate between players and thus have even length. Without a 4-cycle, a minimal cycle has length at least six. 
    A 6-cycle requires at least 3 strategies per player, so the smallest possible size for such a two-player game is $3\times 3$. Now suppose a $3\times 3$ game has a 6-cycle and no sinks. The three non-cycle profiles cannot be comparable, and because each is not a sink it must have at least one arc into another profile, necessarily on the 6-cycle. This gives us, without loss of generality, the black arcs in Figure~\ref{fig:proof shapley}. 

    \begin{figure}
        \centering
        \begin{subfigure}{.4\textwidth}
            \centering
            \begin{tikzpicture}

    \node[nd] (1) at (0,0) {};
    \node[nd] (2) at (1,0) {};
    \node[nd,fill, circle,inner sep=2pt] (3) at (2,0) {};
    \node[nd] (4) at (0,1) {};
    \node[nd,fill,circle,inner sep=2pt] (5) at (1,1) {};
    \node[nd] (6) at (2,1) {};
    \node[nd,fill,circle,inner sep=2pt] (7) at (0,2) {};
    \node[nd] (8) at (1,2) {};
    \node[nd] (9) at (2,2) {};
    
    \draw[arc] (4) to (1);
    \draw[arc] (8) to (9);
    
    \draw[arc,red] (7) to node[label=above] {3} (4);
    \draw[arc, blue] (7) to[in = 130, out = -130] node[label=above] {4} (1);
    \draw[arc] (7) to node[label=above] {1} (8);
    \draw[arc, blue] (7) to[in = 140, out = 40] node[label=above] {2} (9);
    
    \draw[arc] (1) to (2);
    \draw[arc] (2) to[in = -60, out = 60] (8);
    \draw[arc] (9) to (6);
    \draw[arc] (6) to[out = 210, in = -30] (4);

 \end{tikzpicture}
            \caption{Case 1}
            \label{fig:s1}
        \end{subfigure}
        \quad
        \begin{subfigure}{.4\textwidth}
            \centering
            \begin{tikzpicture}

    \node[nd] (1) at (0,0) {};
    \node[nd] (2) at (1,0) {};
    \node[nd,fill,circle,inner sep=2pt] (3) at (2,0) {};
    \node[nd] (4) at (0,1) {};
    \node[nd,fill,circle,inner sep=2pt] (5) at (1,1) {};
    \node[nd] (6) at (2,1) {};
    \node[nd,fill,circle,inner sep=2pt] (7) at (0,2) {};
    \node[nd] (8) at (1,2) {};
    \node[nd] (9) at (2,2) {};
    
    \draw[arc] (4) to (1);
    \draw[arc] (8) to (9);
    
    \draw[arc,red] (7) to node[label=above] {2} (4);
    \draw[arc, blue] (7) to[in = 130, out = -130] node[label=above] {3} (1);
    \draw[arc, red] (7) to node[label=above] {4} (8);
    \draw[arc] (7) to[in = 140, out = 40] node[label=above] {1} (9);
    
    \draw[arc] (1) to (2);
    \draw[arc] (2) to[in = -60, out = 60] (8);
    \draw[arc] (9) to (6);
    \draw[arc] (6) to[out = 210, in = -30] (4);

 \end{tikzpicture}
            \caption{Case 2}
            \label{fig:s2}
        \end{subfigure}
        \caption{The two cases in the proof of Lemma~\ref{shapley uniqueness}, determined by which arc is selected first (labelled 1). Black arcs are by assumption (without loss of generality), blue arcs are determined by transitivity along rows and columns and red arcs are determined by the absence of 4-cycles. The numbers are the order in which we consider each arc.}
        \label{fig:proof shapley}
    \end{figure}
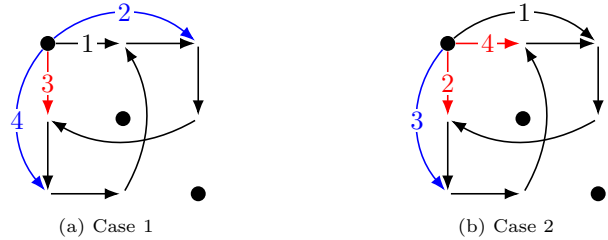

    The proof splits into two very similar cases depending on what arc is chosen first (without loss of generality). We examine the arcs in the order given by the labels in Figures~\ref{fig:s1} and \ref{fig:s2}, and determine in each case that the constraints of transitivity along rows and columns and the absence of 4-cycles forces the arcs to go from the node in the cycle. Repeating this argument for each of the three non-cycle nodes gives us the Shapley graph. As a result, it is the unique $3\times 3$ game with this property, and hence the unique smallest game as any such two-player game needs to be at least $3\times 3$.
\end{proof}

\cite{jordan_three_1993} introduced a similar game to study the non-convergence of FP and related dynamics. This game, sometimes called `mismatching pennies'~\citep{sandholm2010population,kleinberg_beyond_2011}, is a $2\times 2\times 2$ game where players sit in a circle, and each aims to \emph{mis}match the choice of the player to their left. The resultant preference graph is shown in Figure~\ref{fig:jordan}. Like Shapley's game, Jordan's has a unique interior equilibrium, but its three-dimensional strategy space is easier to visualise (it is isomorphic to the unit cube)~\citep[see][Figure 9.6]{sandholm2010population}. Comparing Figures~\ref{fig:shapley} and \ref{fig:jordan}, we see that Jordan's game and Shapley's have a very similar structure---in fact, Jordan's game is the unique smallest game (in terms of strategy profiles, with any number of players) to possess neither a 4-cycle nor a sink.

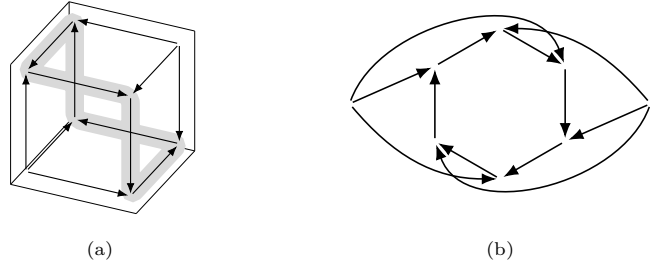
\begin{figure}
    \centering
    \begin{subfigure}{.4\textwidth}
        \centering
        \begin{tikzpicture}[scale=.7]
\begin{axis}[
    axis line style = {opacity=0},
    xmin=-0.1, xmax=1.1,
    ymin=-0.1, ymax=1.1,
    zmin=-0.1, zmax=1.1,
    axis equal image,
    xticklabels = {},
    yticklabels = {},
    zticklabels = {},
    xtick style = {draw=none},
    ytick style = {draw=none},
    ztick style = {draw=none},
]

    \draw [rounded corners, color = gray!30, line width = 10pt] (1,0,1) to (0,0,1) to (0,1,1) to (0,1,0) to (1,1,0) to (1,0,0) -- cycle;

    
    \node[nd] (hhh) at (0,0,0) {};
    \node[nd] (thh) at (1,0,0) {};
    \node[nd] (hht) at (0,0,1) {};
    \node[nd] (tht) at (1,0,1) {};
    
    \node[nd] (hth) at (0,1,0) {};
    \node[nd] (tth) at (1,1,0) {};
    \node[nd] (htt) at (0,1,1) {};
    \node[nd] (ttt) at (1,1,1) {};
    
    
    \draw[arc] (hhh) to (hht);
    \draw[arc] (hhh) to (hth);
    \draw[arc] (hhh) to (thh);
    \draw[arc] (ttt) to (htt);
    \draw[arc] (ttt) to (tth);
    \draw[arc] (ttt) to (tht);
    \draw[arc] (htt) to (hht);
    \draw[arc] (hht) to (tht);
    \draw[arc] (hth) to (htt);
    \draw[arc] (tth) to (hth);
    \draw[arc] (thh) to (tth);
    \draw[arc] (tht) to (thh);
\end{axis}
 \end{tikzpicture}
        \caption{}
        \label{fig:jordan-square}
    \end{subfigure}
    \quad
    \begin{subfigure}{.4\textwidth}
        \centering
        \begin{tikzpicture}

    \node[nd] (1) at (0,0) {};
    \node[nd] (2) at (-0.866,0.5) {};
    \node[nd] (3) at (-0.866,1.5) {};
    \node[nd] (5) at (0.866,1.5) {};
    \node[nd] (4) at (0,2) {};
    \node[nd] (6) at (0.866,0.5) {};
    
    \draw[arc] (1) to (2);
    \draw[arc] (2) to (3);
    \draw[arc] (3) to (4);
    \draw[arc] (4) to (5);
    \draw[arc] (5) to (6);
    \draw[arc] (6) to (1);
    
    \node[nd] (7) at (-2,1) {};
    \node[nd] (8) at (2,1) {};
    
    \draw[arc] (7) to[out=-50, in = 180] (1);
    \draw[arc] (7) to (3);
    \draw[arc] (7) to[out = 70, in = 110] (5);
    
    \draw[arc] (8) to[out = -110, in = -70] (2);
    \draw[arc] (8) to[out = 130, in = 0] (4);
    \draw[arc] (8) to (6);
 \end{tikzpicture}
        \caption{}
        \label{fig:jordan-cycle}
    \end{subfigure}
    \caption{Two drawings of Jordan's game `Mismatching Pennies' \citep{jordan_three_1993}. Figure~\ref{fig:jordan-square} reflects the payoff matrix structure, while Figure~\ref{fig:jordan} more clearly shows the graph structure. The graph has two sources with arcs into a single sink equilibrium, which is a 6-cycle, highlighted in grey.}
    \label{fig:jordan}
\end{figure}

\begin{lem} \label{jordan uniquenesss}
    Jordan's graph is the unique game with fewest profiles with neither a sink nor 4-cycle.
\end{lem}
\begin{proof}
    We know that Shapley's game is the minimal two-player game with this property, and it has 9 profiles. Any smaller game can have at most 8, and must have at least three players---thus this game must be $2\times 2\times 2$. All three players must participate in the cycle, so it must have even length. Note that an 8-cycle cannot be minimal in this game, as some arcs would `shortcut' the cycle. Thus we seek a $2\times 2\times 2$ game containing a 6-cycle and no sinks. There are two non-cycle profiles, which must be diagonally opposite each other. Each is not a sink, so must have at least one arc into another profile, necessarily on the 6-cycle. This gives us the black arcs in Figure~\ref{fig:proof jordan}. 

    \begin{figure}
        \centering
        \begin{tikzpicture}[scale=.7]
\begin{axis}[
    axis line style = {opacity=0},
    xmin=-0.1, xmax=1.1,
    ymin=-0.1, ymax=1.1,
    zmin=-0.1, zmax=1.1,
    axis equal image,
    xticklabels = {},
    yticklabels = {},
    zticklabels = {},
    xtick style = {draw=none},
    ytick style = {draw=none},
    ztick style = {draw=none},
]


    
    \node[nd,circle,fill, inner sep=2pt] (hhh) at (0,0,0) {};
    \node[nd] (thh) at (1,0,0) {};
    \node[nd] (hht) at (0,0,1) {};
    \node[nd] (tht) at (1,0,1) {};
    
    \node[nd] (hth) at (0,1,0) {};
    \node[nd] (tth) at (1,1,0) {};
    \node[nd] (htt) at (0,1,1) {};
    \node[nd,circle,fill, inner sep=2pt] (ttt) at (1,1,1) {};
    
    
    \draw[arc, red] (hhh) to node[label=above] {2} (hht);
    \draw[arc] (hhh) to node[label=above] {1} (hth);
    \draw[arc, red] (hhh) to node[label=above] {3} (thh);
    \draw[arc] (htt) to (hht);
    \draw[arc] (hht) to (tht);
    \draw[arc] (hth) to (htt);
    \draw[arc] (tth) to (hth);
    \draw[arc] (thh) to (tth);
    \draw[arc] (tht) to (thh);
\end{axis}
 \end{tikzpicture}
        \caption{Once one arc (black) is fixed, the remaining arcs (red) are forced by the absence of 4-cycles. The numbers are the order in which we consider each arc.}
        \label{fig:proof jordan}
    \end{figure}
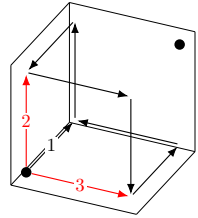
    Again, we find that the absence of 4-cycles forces the remaining arcs to go out of this node into the 6-cycle. Repeating this argument for both non-cycle nodes gives us the Jordan graph.
\end{proof}

\citeauthor{jordan_three_1993} proved convergence to the sink equilibrium (the 6-cycle) for a particular choice of payoffs in this game\footnote{Specifically, each player received a payoff of 1 for matching and 0 for mismatching.}. This was generalised by \cite{gaunersdorfer_fictitious_1995} who showed that FP converges to a 6-cycle for much more general choices of parameters, demonstrating that the limit behaviour is determined by the graph rather than the payoffs.

The absence of 4-cycles appears to explain non-convergence of FP in Shapley's and Jordan's games. This has a surprising echoes from the case of zero-sum games, where fictitious play does converge---there, every non-dominance-solvable game always has a 4-cycle. More strongly, \emph{every strategy} must participate in at least one 4-cycle. This follows from Shapley's theorem \citep{shapley_topics_1964} (see Section~\ref{sec: zero-sum}, and Theorem~\ref{dominance theorem} in particular).


Note that cycle structure appears to be important even beyond the sink equilibria. For instance, the sink equilibria of ordinal potential and weakly acyclic games are the same (they are all PNEs) yet the convergence behaviour can differ because weakly acyclic games have cycles elsewhere. \cite{foster_nonconvergence_1998} demonstrate a very interesting example of a $9\times 9$ game which is weakly acyclic, yet from certain initial conditions the game converges to an 18-cycle in the preference graph. This also shows that the sink equilibria do not capture all long-run outcomes of FP in large games, because this cycle is not a sink.

A complete understanding of the convergence of FP in general games remains unknown. However, a theme clearly emerges from this story: with increasing understanding of FP comes an increasing understanding of the influence of the preference graph. In particular, graph properties like connectivity and cycles lie at the heart of what is known about this dynamic.

\subsection{Replicator Dynamic} \label{sec: replicator}

In evolutionary game theory, the flagship dynamic is the \emph{replicator dynamic}, a conceptually simple and well-motivated model of co-evolution \citep{hofbauer_evolutionary_2003,sandholm2010population}. This model was introduced in \cite{smith1973logic}, the landmark study which created the field of evolutionary game theory. It was formalised and named by \cite{taylor_evolutionary_1978}, and has since retained and strengthened its central role in the field of game dynamics.
\begin{defn}[Replicator dynamic \citep{taylor_evolutionary_1978}]
    The replicator dynamic is the solution of the following ordinary differential equation:
    \begin{equation}
\dot x_s^p = x_s^p\big (\exptu_p(s; x_{-p}) - \sum_{t\in S_p} x_t^p \exptu_p(t; x_{-p}) \big )
\end{equation}
\end{defn}



Unlike fictitious play, the replicator was motivated by evolutionary biology, but it turns out to also have an intimate connection with learning: the replicator is the continuous-time limit of the \emph{multiplicative weight update method}, which lies at the heart of modern online learning techniques \citep{arora_multiplicative_2012}. Compared to FP, the replicator can be viewed as an application of the ML technique of \emph{regularisation}---instead of selecting the best strategy at each step, the players select a mixed strategy which is a weighted approximation of the best strategy. Specifically, it is the \emph{logit} \citep{luce1959individual} of the best response, also called the `\emph{softmax}'~\citep{boyd2004convex}. Like FP, the study of the replicator has a long history of convergence results. PNEs are attractors of the replicator \citep{sandholm2010population}, and we converge to them in many classes of games where they are guaranteed to exist, including potential \citep{monderer_potential_1996}, supermodular \citep{milgrom_rationalizability_1990} and dominance-solvable games~\citep{sandholm2010population}. More generally, whenever a subgame is closed under `weakly better responses' it is asymptotically stable under the replicator \citep{gaunersdorfer_fictitious_1995}; in preference graph terminology, this occurs exactly when a sink equilibrium is a subgame, generalising the PNE case. Actually, this is the only case: a subgame is asymptotically stable under the replicator if and only if it is closed in the preference graph~\citep{biggar_replicator_2023}. Unlike FP, the replicator doesn't converge to equilibria in zero-sum games. However, the trajectories of FP in \emph{belief space} are implicitly averaged over time. When we take the \emph{time-average} of the replicator, it closely approximates FP \citep{gaunersdorfer_fictitious_1995,hofbauer_time_2009,viossat_no-regret_2013}. Consequently, the time-average of the replicator dynamic also converges to equilibrium in zero-sum games. This generalises to an important theorems in online learning: any no-regret learning algorithm converges to equilibrium in a zero-sum game~\citep{cesa2006prediction}.

Like FP, a connection between the replicator and the preference graph does not seem obvious. In fact, at first glance, the replicator appears even farther from the preference graph than FP, because the replicator evolves in mixed strategy space, and its trajectories are clearly highly dependent on the cardinal values of the payoffs. However, contrary to this intuition, the preference graph turns out to also have a strong influence on the long-run behaviour of the replicator. For instance, PNEs are attractors of the replicator, and these are also sink equilibria, which demonstrates converges in the classes of games which PNEs (e.g. (ordinal) potential, (quasi-)supermodular and dominance-solvable games). When sink equilibria and Nash equilibria do not coincide, it is the sink equilibrium which appears to predict convergence of the replicator. For instance, in Shapley's and Jordan's game, we converge to the sink equilibrium rather than the Nash equilibrium \citep{kleinberg_beyond_2011,viossat_no-regret_2013}. \cite{biggar_replicator_2023} proved that every attractor of the replicator contains a sink equilibrium. As Shapley's and Jordan's games demonstrate, the same statement doesn't apply to Nash equilibria. Even in zero-sum games, an inherently cardinal class, the sink equilibria were found to totally characterise the unique attractor of the replicator \citep{biggar_attractor_2024}. The uniqueness has an easy graph-theoretic proof: all zero-sum games have a unique sink equilibrium \citep{biggar_graph_2023}, and each attractor contains a sink equilibrium. The message, it seems, is that the sink equilibria---as opposed to the Nash equilibria---appear to govern the long-run behaviour of the replicator dynamic.

\subsection{Markov game dynamics and the Price of Sinking} \label{sec: markov}

Convergence alone is not always sufficient---for some applications, we would like to additionally know which possible solution concepts are reached, and with what frequency. In other words, we wish to be able to perform \emph{equilibrium selection}~\citep{harsanyi1988general} and \emph{strategic ranking}~\citep{balduzzi_re-evaluating_2018,omidshafiei_-rank_2019}. These are especially relevant in algorithmic game theory applications, such as computing the \emph{Price of Anarchy} of a game \citep{roughgarden2010algorithmic}. Unfortunately, these can be difficult to compute~\citep{goldberg_complexity_2013}. Luckily, however, a good approximation is often sufficient. This is one of the original motivations for the study of the preference graph. 

The Price of Anarchy (PoA) is the ratio of the optimal social welfare to that achieved in the worst Nash equilibrium~\citep{koutsoupias1999worst}. The concept is typically studied on congestion games. Because congestion games are potential games \citep{monderer_potential_1996}, PNEs exist and coincide with sink equilibria. However, in non-congestion games computing the PoA is difficult~\citep{daskalakis_complexity_2009}. The sink equilibria were presented as an alternative~\citep{goemans_sink_2005,mirrokni_convergence_2004}, as a generalisation of PNEs which is easy to compute and captures convergence of simple dynamics. The resultant value---called the `Price of Sinking' \citep{goemans_sink_2005}---can be significantly different from the Price of Anarchy. \cite{kleinberg_beyond_2011} demonstrated that, on Jordan's game, the replicator dynamic converges to the sink equilibrium and the social welfare on that 6-cycle can be arbitrarily higher than that of the unique Nash equilibrium, giving a dramatically different value for the Price of Anarchy.

A related algorithmic game theory problem concerns \emph{strategic ranking}: assigning a `strength' to different strategies in a game. Aside from its game-theoretic motivation, this is an increasingly important problem in multi-agent learning \citep{balduzzi_re-evaluating_2018}, where we need to assign reward to different agents. Because the preference graph is a simple approximation of `general' dynamics on a game, it has also been applied for this purpose. An important practical example is the \emph{$\alpha$-rank} method \citep{omidshafiei_-rank_2019}.


However, in order to use sink equilibria to compute either the PoA or a strategic ranking, we must assign a numerical value to each profile, representing the `proportion' it is selected. A simple approach, often rediscovered \citep{papadimitriou_game_2019,goemans_sink_2005,young_evolution_1993,fabrikant2008complexity,omidshafiei_-rank_2019,hakim2024swim}, is to use a Markov chain on the strategy profiles to approximate dynamics on the game. In fact, this idea goes back as far as \cite{cournot1838recherches}\footnote{Cournot described a process for reaching an equilibrium where each player plays a best-response to their opponent's strategy at the last time step~\citep{moulin_dominance_1984}. This process is sometimes called \emph{Cournot tatonnement}---it is easy to see it is exactly a random walk on the best-response graph.}. Typically, a specific form of Markov chain is used, where transitions only occur between comparable profiles which increase the payoff for the deviating player. We use the term \emph{Markov game dynamic} to refer to any game dynamic defined by such a Markov chain.
\begin{defn}[Markov game dynamics]
    Let $u$ be the utility function of an $N$-player game, with $Z = \prod_{i=1}^N S_i$ the set of profiles. Let $f$ be a mapping from such a game $u$ to a Markov transition matrix $M_u$ on the profiles in $Z$. Thus $f$ defines a dynamical system on a game through the associated Markov chain of $M_u$. We say that $f$ is a \emph{Markov game dynamic} if for any $u$ and $p,q\in Z$, $(M_u)_{p,q} > 0$ implies that $p$ and $q$ are $i$-comparable and $u_i(p) \leq u_i(q)$.
\end{defn}
That is, a Markov game dynamic is one which given any game $u$, the transition graph of the Markov chain defined by $f(u)$ is a subgraph of the preference graph\footnote{Called a \emph{Markov-Conley chain} in \cite{papadimitriou_game_2019}.}. The following result is immediate:
\begin{lem} \label{lem: stationary distributions}
    If $f$ is a Markov game dynamic, then for any game $u$ the stationary distribution of $f(u)$ is contained in some sink equilibrium of $u$.
\end{lem}

There are two important subclasses of Markov game dynamics: \emph{preference-Markov game dynamics}, where $(M_u)_{p,q} > 0$ \emph{if and only if} that $p$ and $q$ are $i$-comparable and $u_i(p) \leq u_i(q)$; and \emph{best-response-Markov game dynamics} where $(M_u)_{p,q} > 0$ \emph{if and only if} that $p$ and $q$ are $i$-comparable and $q_i$ is a best-response to $p_{-i}$. That is, a Markov game dynamic is preference-Markov if the transition graph of $f(u)$ is exactly the preference graph of $u$, and is best-response-Markov if the transition graph of $f(u)$ is the best-response graph of $u$. Examples of preference-Markov game dynamics include those from \citep{papadimitriou_game_2019,omidshafiei_-rank_2019,hakim2024swim}, while some example of best-response-Markov game dynamics include \citep{cournot1838recherches,goemans_sink_2005,fabrikant2008complexity}.

In preference-Markov and best-response-Markov game dynamics, Lemma~\ref{lem: stationary distributions} can be significantly strengthened:
\begin{lem} \label{markov dynamic lemma}
If $f$ is a preference-Markov (resp. best-response-Markov) game dynamic, then for any $u$, (1) each sink equilibrium (resp. best-response sink equilibrium) has a unique stationary distribution whose support is that equilibrium and (2) the set of stationary distributions is exactly the set of mixtures of these distributions.
\end{lem}
In other words, the sink equilibria characterise the qualitative properties of the stationary distributions of these dynamics. Lemma~\ref{markov dynamic lemma} follows from standard Markov chain properties \citep{kemeny1969finite}. Once we have selected a Markov game dynamic, we can easily find a distribution over sink equilibria from any given prior distribution over the profiles. 
However, while the \emph{qualitative} properties of the stationary distributions (\emph{i.e.} the support) of preference-Markov game dynamics depend only on the preference graph, the numerical values of the distribution also depend on the specific transition probabilities of the Markov chain, and the prior distribution over strategies. That is while the qualitative behaviour is the same over all preference-Markov game dynamics, quantitative values depend on the particular dynamic chosen. This leads to a problem---there is not necessarily a canonical way of assigning numerical transition probabilities to arcs in the preference graph. Note the similarity of this problem to the general problem of modelling numerical payoffs, which was one of the motivating ideas for the preference graph in the first place.

One plausible choice of preference-Markov game dynamic \citep{papadimitriou_game_2019,hakim2024swim} is to assign the transition probability on an arc to be the (normalised) weight of that arc in the weighted preference graph (Definition~\ref{def: weighted preference})\footnote{If arcs can have zero weight, then computing this long-run distribution of a Markov chain is highly complex. However, \cite{hakim2024swim} recently showed that we can compute the sink equilibrium distribution efficiently even in non-strict preference graphs.}. However, more study is needed to determine if this is truly the `canonical' preference-Markov game dynamic, which well-approximates the qualitative behaviour of more-complicated dynamics.

By contrast, for two-player games, all best-response-Markov game dynamics (such as Cournot tatonnement) have the same stationary distributions. This means we can obtain these distributions without needing to find a `canonical' choice of dynamic. The reason for this is that every sink equilibrium of the two-player best-response graph is a cycle, and there is a unique choice of Markov chain whose transition graph is a cycle. However, the distribution over the sink equilibria may still depend on the transition probabilities on non-sink profiles. In addition, we know from Section~\ref{sec: better and best} that the best-response graph is somewhat limited in its ability to approximate general dynamics. See Section~\ref{sec: conclusions}.




\subsection{General dynamics}\label{sec: other dynamics}

Fictitious play, the replicator and Markov game dynamics are three interesting classes of game dynamics which have suggestive connections with the preference graph. However, they are not all dynamics. Should we expect these connections to generalise?

We believe so, with the following intuition. Very broadly speaking, the preference graph is a discrete representation of the \emph{gradient} of the utility function in the neighbourhood of every pure profile. The arcs in and out of a pure profile in the preference graph represent directions of increasing or decreasing payoff in the strategy space. Traversing the preference graph is therefore a representation of applying \emph{gradient descent} to a game, possibly with some regularisation, and many game dynamics are directly or indirectly built on this principle~\citep{mertikopoulos2016learning}. The preference graph is a `skeleton' \citep{biggar_graph_2023,biggar_replicator_2023,pangallo_best_2019}, which approximates the dynamics in regions where all-but-one players play a nearly-pure strategy.

For instance, fictitious play and the replicator are both examples of a class of learning-inspired dynamics called `Follow-The-Regularised-Leader' (FTRL) \citep{shalev2012online}. These dynamics select, at each step, the best strategy for each player, subject to some kind of regularisation. \cite{vlatakis-gkaragkounis_no-regret_2020} proved that mixed Nash equilibria cannot be (last-iterate) stable under FTRL, and that instead any attracting set must contain a pure profile. This is extended by the theorem of \cite{biggar_replicator_2023}, from which it follows that for any \emph{subgame-invariant}\footnote{The behaviour on a subgame is independent of the rest of the game; the replicator is one example.} FTRL algorithm, every asymptotically stable set contains a sink equilibrium. 

A different class of dynamics are motivated by computationally limited players. At the extreme end of this spectrum lie Markov game dynamics (Section~\ref{sec: markov}). An intermediate approach is taken by \cite{young_evolution_1993}, who defined a dynamic called `adaptive play' where players sample from a bounded number of previous rounds of play before choosing a best-response. Young proved that adaptive play converges to PNEs in all weakly acyclic games.

The ultimate goal of game dynamics is to understand the behaviour of real-world decision-makers. For these agents, it is essentially impossible to know if a particular dynamic is `actually' employed. Further, the payoffs in games cannot be precisely known, and even simple dynamics like the replicator exhibit chaotic behaviour in small games \citep{sato_chaos_2002}\footnote{Even zero-sum games, where the time-average convergence is well-established, the day-to-day behaviour can be essentially unpredictable \citep{cheung2019vortices,papadimitriou_nash_2018}. In response to the challenge of chaotic behaviour, \citeauthor{papadimitriou_game_2019} proposed the \emph{sink chain components}---a robust dynamical systems concept similar to an attractor---as the outcome of game dynamics. While chain components might be complex objects, they suggested that we use the sink equilibria of the preference graph to approximate them. It turns out that every sink chain component of the replicator contains a sink equilibrium \citep{biggar_replicator_2023}.}.
The preference graph and its sink equilibria do not immediately solve these problems, but they present a much more positive picture than the Nash equilibria~\citep{milgrom_adaptive_1991,papadimitriou_game_2019}. The preference graph is robust to payoff changes, explains convergence and non-convergence in many games, and its sink equilibria are an easily-computable approximation of the limiting behaviour of some dynamics. The extent to which the preference graph explains the behaviour of general dynamics is a fundamental question for future work.

\section{Conclusions and Future Work} \label{sec: conclusions}

In this paper we examined the preference graph, and surveyed its influence and potential for the field of game theory. We found that the preference graph can explain many game-theoretic results, both classical and modern, in a payoff-robust way. We believe that understanding the preference graph could lead to a significantly different theory of games, especially for game dynamics.

Unsurprisingly, for such a simple yet understudied object, there are far more questions to investigate. Here we will present an outline of some broad research questions raised by the preference graph.

\begin{itemize}
    \item (General convergence of dynamics) Fundamentally, we wish to understand the behaviour of real players in real games. We believe that an answer to this question must necessarily involve dynamics/learning~\citep{papadimitriou_game_2019}. The preference graph can explain many established game dynamics results \citep{shapley_topics_1964,jordan_three_1993,berger_two_2007} and has provided paths to new solutions as well~\citep{biggar_replicator_2023,biggar_attractor_2024}. How far can this take us? Is the sink equilibrium (or some variant thereof, see below) a new `truly predictive' solution concept for game theory?
    \item (Distributions over outcomes) Like Nash equilibria, many sink equilibria exist. This means that the problem of equilibrium selection carries over from the Nash equilibrium case \citep{harsanyi1988general}. As we discussed in Section~\ref{sec: markov}, dynamics presents a solution to equilibrium selection by defining a distribution over sink equilibria, given some prior over initial beliefs~\citep{papadimitriou_game_2019}. However, we do not know how to compute this for general dynamics, or even a canonical choice of approximation by a Markov game dynamic. This is important to make predictions from normal-form games.
    \item (Other equilibrium concepts) What is the general relationship between the preference graph/sink equilibria and equilibrium concepts such as Nash equilibria, correlated equilibria and others? Can this offer insight into stability and computability of these equilibrium notions?
    \item (Refining the sink equilibrium) The sink equilibrium is a natural and promising candidate solution concept. However, we know of several instances where its prediction is weaker than expected. The sink equilibrium can contain dominated strategies (Section~\ref{sec:dominance}); under FP some other cycles in the preference graph can be stable (\cite{foster_nonconvergence_1998}, Section~\ref{sec: FP}); and in zero-sum games, the support of the Nash equilibrium can be a strict subset of the sink equilibrium~\citep{biggar_attractor_2024}. However, these limitations seem possibly fixable---particularly the first, since dominated strategies are also present in the preference graph. Overall, it seems like some refinement is needed if the sink equilibrium is to provide a general solution to game dynamics.
    \item (Partially-ordered preferences) The preference graph allows us to apply game theory to agents who possess only partially-ordered preferences. It would be interesting to understand what new concepts and questions arise from this model.
\end{itemize}

\bibliographystyle{elsarticle-harv}
\bibliography{refs}

\end{document}